\documentclass{llncs}

\usepackage{latexsym, amsmath, amssymb}
\usepackage{url}
\usepackage{pifont} % for \ding

\usepackage{tikz}
\usetikzlibrary{shapes,snakes,arrows,positioning,calc}

\usepackage{booktabs}
\usepackage{color, colortbl}

\renewcommand{\qed}{\hfill\ding{113}}

\newcommand{\PTL}{\ensuremath{\mathsf{LTL}}}
\newcommand{\coreLTL}{\ensuremath{\PTL_{\textit{core}}}}
\newcommand{\kromLTL}{\ensuremath{\PTL_{\textit{krom}}}}
\newcommand{\boolLTL}{\ensuremath{\PTL_{\textit{bool}}}}
\newcommand{\hornLTL}{\ensuremath{\PTL_{\textit{horn}}}}

\newcommand{\Xallop}{^{\Box,\raisebox{1pt}{$\scriptscriptstyle\bigcirc$}}}
\newcommand{\Xbox}{^\Box}
\newcommand{\Xu}{^{\ooalign{\hss\raisebox{1pt}{$\scriptscriptstyle*$}\hss\cr$\scriptstyle\Box$}}}

\newcommand{\iref}[1]{\textup{\bfseries{#1}}}

\newcommand{\NLogSpace}{\textsc{NLogSpace}}
\newcommand{\PTime}{\textsc{PTime}}
\newcommand{\NP}{\textsc{NP}}
\newcommand{\PSpace}{\textsc{PSpace}}

\newcommand{\SVbox}{\mathop{\ooalign{$\Box$ \cr \kern0.42ex
    \raisebox{0.35ex}{\scalebox{0.7}{$*$}}}\rule{0pt}{1.5ex} \kern-0.7ex}}
    
\newcommand{\SVdiamond}{\mathop{\ooalign{$\Diamond$ \cr \kern0.5ex
    \raisebox{0.35ex}{\scalebox{0.7}{$*$}}} \kern-0.9ex}}

\newcommand{\nxt}{{\ensuremath\raisebox{0.25ex}{\text{\scriptsize$\bigcirc$}}}}
\newcommand{\Rnext}{\nxt_{\!\scriptscriptstyle F}}
\newcommand{\Lnext}{\nxt_{\!\scriptscriptstyle P}}

\newcommand{\Rdiamond}{\Diamond_{\!\scriptscriptstyle F}}
\newcommand{\Ldiamond}{\Diamond_{\!\scriptscriptstyle P}}

\newcommand{\Rbox}{\Box_{\!\scriptscriptstyle F}}
\newcommand{\Lbox}{\Box_{\!\scriptscriptstyle P}}

\newcommand{\Until}{\mathbin{\mathcal{U}}}
\newcommand{\Since}{\mathbin{\mathcal{S}}}

\newcommand{\K}{\mathfrak{K}_\varphi}
\newcommand{\M}{\mathfrak{M}}
\newcommand{\Z}{\mathbb{Z}}
\newcommand{\der}{\mathfrak{d}}

\newcommand{\PP}{{\scriptscriptstyle P}}
\newcommand{\FF}{{\scriptscriptstyle F}}

\title{The Complexity of Clausal Fragments of LTL}

\author{A.~Artale,\!$^1$ R.~Kontchakov,\!$^2$ V.~Ryzhikov,\!$^1$ and  M.~Zakharyaschev$^2$}

\institute{
 \begin{minipage}[t]{52mm} \centering
   $^1$ KRDB Research Centre\\
   Free University of Bozen-Bolzano\\
   I-39100 Bolzano, Italy\\
   \texttt{\{artale,ryzhikov\}@inf.unibz.it}
 \end{minipage}
 \hfill
 \begin{minipage}[t]{68mm} \centering
   $^2$ Dept.\ of Computer Science and Inf.~Systems\\
   Birkbeck, University of London\\
   London WC1E 7HX, UK\\
   \texttt{\{roman,michael\}@dcs.bbk.ac.uk}
 \end{minipage}
}

\authorrunning{A. Artale, R. Kontchakov, V. Ryzhikov and M. Zakharyaschev}

\begin{document}

\maketitle

\begin{abstract}

We introduce and investigate a number of fragments of propositional temporal logic \PTL{} over the flow of time $(\Z,<)$. The fragments are defined in terms of the available temporal operators and the structure of the clausal normal form of the temporal formulas. We determine the computational complexity of the satisfiability problem  for each of the fragments, which ranges from \NLogSpace{} to \PTime, \NP{} and \PSpace.

\end{abstract}

\section{Introduction}\label{sec:2ltl}

We consider the (\PSpace-complete) propositional temporal logic \PTL{} over the flow of time $(\mathbb Z, <)$. Our aim is to investigate how the computational complexity of the satisfiability problem for \PTL-formulas depends on the form of their clausal representation and the available temporal operators. 

Sistla and Clarke~\cite{SistlaClarke82} showed that satisfiability of \PTL{}-formulas with all standard operators (`next-time', `always in the future', `eventually' and `until') is \PSpace-complete; see also~\cite{HalpernR81,DBLP:conf/lop/LichtensteinPZ85}. Ono and Nakamura~\cite{OnoNakamura80} proved that for formulas with only `always in the future' and `eventually' the satisfiability problem becomes \NP-complete. 
Since then a number of fragments of \PTL{} of different complexity have been identified.
For example, Chen and Lin~\cite{ChenLin93} observed that the complexity  
does not change if we restrict attention to  temporal Horn formulas. Demri and Schnoebelen~\cite{DBLP:journals/iandc/DemriS02} determined the complexity of fragments that depend on three parameters: the available temporal operators, the number of nested temporal operators, and the number of propositional variables in formulas. Markey~\cite{Markey04} analysed fragments defined  by the allowed set of temporal operators, their nesting and the use of negation.
Dixon \emph{et al.}~\cite{DBLP:conf/ijcai/DixonFK07} % DixonFK06,
introduced a XOR fragment of \PTL{} and showed its tractability. Bauland~\emph{et~al.}~\cite{DBLP:journals/corr/abs-0812-4848} systematically investigated the complexity of fragments given by  both temporal operators and Boolean connectives (using Post's lattice of sets of Boolean functions).

In this paper, we classify temporal formulas according to their clausal normal form. Recall~\cite{DBLP:conf/ijcai/Fisher91} that any \PTL-formula over $(\mathbb{N},<)$ can be transformed into an equisatisfiable formula in the so-called \emph{separated normal form} that consists of initial clauses (setting  
conditions at moment 0), step clauses (defining transitions  between consecutive states), and eventuality clauses  (defining the states that must be reached infinitely often). Our clausal normal form is a slight generalisation of the separated normal form. 
The main building blocks are \emph{positive temporal literals} $\lambda$ given by the following grammar:
\begin{equation}\label{lambda}
\lambda \ \ ::= \ \ \bot \ \ \mid \ \ p   \ \ \mid \ \  \Rnext \lambda \ \ \mid \ \  \Lnext \lambda  \ \ \mid \ \  \Rbox \lambda \ \ \mid \ \  \Lbox \lambda \ \ \mid \ \  \SVbox \lambda,
\end{equation} 
where $p$ is a propositional variable, $\Rnext$ and $\Lnext$ are the next- and previous-time operators, and $\Rbox$, $\Lbox$, $\SVbox$  are the operators `always in the future,\!' `always in the past' and  %is the operator 
`always.\!' We say that a temporal formula $\varphi$ is in \emph{clausal normal form} if
\begin{equation}\label{normal}
\varphi \ \  ::= \ \ \lambda \ \ \mid \ \ \neg \lambda  \ \ \mid \ \  \SVbox (\neg \lambda_1 \lor \dots \lor \neg \lambda_n \lor
\lambda_{n+1} \lor \dots \lor \lambda_{n+m})  \ \ \mid \ \  \varphi_1 \land \varphi_2. %, \quad m,n\ge 0. 
\end{equation}
Conjunctions of positive and \emph{negative} ($\neg\lambda$) literals can be thought of as initial clauses, while conjunctions of $\SVbox$-formulas generalise both step and eventuality clauses of the separated normal form. Similarly to~\cite{DBLP:journals/tocl/FisherDP01} one can show  that any \PTL-formula over $(\Z,<)$ is equisatisfiable to a formula in clausal normal form.

We consider twelve fragments of \PTL{} that will be denoted by $\PTL_\alpha\Xallop$, $\PTL_\alpha\Xbox$ and $\PTL_\alpha\Xu$, for $\alpha \in \{\textit{bool},\textit{horn},\textit{krom},\textit{core}\}$. The superscript in the language name indicates the temporal operators that can be used in its positive literals. Thus, $\PTL_\alpha\Xallop$ uses all types of positive literals, $\PTL_\alpha\Xbox$ can only use  the $\Box$-operators: 
\begin{equation*}
\lambda \ \ ::= \ \ \bot \ \ \mid \ \ p   \ \ \mid \ \  %\Rnext \lambda \ \ \mid \ \  \Lnext \lambda  \ \ \mid \ \  
\Rbox \lambda \ \ \mid \ \  \Lbox \lambda \ \ \mid \ \  \SVbox \lambda,
\end{equation*}
and $\PTL_\alpha\Xu$ only the $\SVbox$-operator:
\begin{equation*}
\lambda \ \ ::= \ \ \bot \ \ \mid \ \ p   \ \ \mid \ \  %\Rnext \lambda \ \ \mid \ \  \Lnext \lambda  \ \ \mid \ \  
%\Rbox \lambda \ \ \mid \ \  \Lbox \lambda \ \ \mid \ \  
\SVbox \lambda. 
\end{equation*}
The subscript $\alpha$ in the language name refers to the form of the clauses 
\begin{equation}\label{clause}
\neg \lambda_1 \lor \dots \lor \neg \lambda_n \lor
\lambda_{n+1} \lor \dots \lor \lambda_{n+m} 
\end{equation}
($m,n\ge 0$) that can be used in the formulas $\varphi$:
\begin{itemize}
\item \textit{bool}-clauses are arbitrary clauses of the form~\eqref{clause},
\item \textit{horn}-clauses have at most one positive literal (that is, $m\leq 1$),
\item \textit{krom}-clauses are binary (that is, $n + m\leq 2$),
\item \textit{core}-clauses are binary with at most one positive literal ($n + m\leq 2$, $m \leq 1$). 
\end{itemize}
The tight complexity bounds in Table~\ref{table:PTL:languages}  show how the complexity of the
satisfiability problem for \PTL-formulas depends on the %type of the underlying propositional 
form of clauses and the available temporal operators.
\begin{table}[ht]
\centering\renewcommand{\arraystretch}{1.5}
\begin{tabular}{cccc}\toprule
temporal operators & $\SVbox, \Rbox, \Lbox$, $\raisebox{2pt}{\tiny$\bigcirc$}\hspace{-0.15em}_{\scriptscriptstyle F}, \raisebox{2pt}{\tiny$\bigcirc$}\hspace{-0.15em}_{\scriptscriptstyle P}$ &
%$\SVbox$, $\raisebox{2pt}{\tiny$\bigcirc$}\hspace{-0.15em}_{\scriptscriptstyle F}/\raisebox{2pt}{\tiny$\bigcirc$}\hspace{-0.15em}_{\scriptscriptstyle P}$ & 
$\SVbox, \Rbox, \Lbox$ & $\SVbox$  \\%\hline 
$\alpha$ & $\PTL_\alpha\Xallop$ & 
%$\PTL_\alpha^{\ooalign{\hss\raisebox{1pt}{$\scriptscriptstyle*$}\hss\cr$\scriptstyle\Box$},\raisebox{1pt}{$\scriptscriptstyle\bigcirc$}}$ & 
$\PTL_\alpha\Xbox$ & $\PTL_\alpha\Xu$ \\\midrule
$\textit{bool}$ & \PSpace{} {\scriptsize($\leq$~\cite{SistlaClarke82})}  %&\PSpace 
& \NP{}  & \NP \\%\hline
$\textit{horn}$ & \PSpace{}  {\scriptsize($\geq$~\cite{ChenLin93})} %&\PSpace{}  
& \PTime{} {\scriptsize $[\leq$ Th.~\ref{thm:hornLTL}$]$} & \PTime \\%\hline
$\textit{krom}$ & \NP{} {\scriptsize $[\leq$ Th.~\ref{lem:bin-ltl:krom-diamond-next-np}$]$} %& \NP 
& \NP{} {\scriptsize $[\geq$ Th.~\ref{krom-low-NP}$]$} & \NLogSpace{}  {\scriptsize $[\leq$ Th.~\ref{thm:kromLTL}$]$} \\%\hline
$\textit{core}$ %& \NP 
& \NP{} {\scriptsize $[\geq$ Th.~\ref{lem:bin-ltl:core-diamond}$]$} & \NLogSpace{} {\scriptsize $[\le$Th.~\ref{newstuff}$]$} & \NLogSpace \\\bottomrule
\end{tabular}\\[6pt]
\caption{The complexity of clausal fragments of \PTL.}
\label{table:PTL:languages}
\end{table}
The \PSpace{} upper bound for $\boolLTL\Xallop$ is
well-known~\cite{HalpernR81,SistlaClarke82,Rabi:10,Rey:10}; the
matching lower bound can be obtained already for $\hornLTL\Xallop$
without $\Rbox$ and $\Lbox$ by a standard encoding of deterministic
Turing machines with polynomial tape~\cite{ChenLin93}.  The \NP{}
upper bound for $\boolLTL\Xbox$ is also well-known~\cite{OnoNakamura80}, and the \PTime{} and \NLogSpace{} lower bounds for $\hornLTL\Xu$ and
$\coreLTL\Xu$ coincide with the complexity of the respective non-temporal languages. 

The main contributions of this paper are the remaining complexity
results in Table~\ref{table:PTL:languages}.  Note first that the
complexity of the $\PTL_\alpha\Xu$ fragments coincides with that of
the underlying propositional fragments.  The complexity of the
$\PTL_\alpha\Xbox$ fragments matches the complexity of the underlying non-temporal fragments except for the Krom case,
where we can use the clauses $\neg p \lor \neg \Rbox q$ and $q \lor r$
to say that $p \to \Rdiamond r$ (if $p$ then eventually $r$), which allows one to encode 3-colourability and results in \NP-hardness. It is known that the addition of
the operators $\Rnext$ and $\Lnext$ to the language with $\Rbox$ and
$\Lbox$ usually increases the complexity (note that the proofs of the
lower bounds for the $\PTL_\alpha\Xallop$ fragments require only
$\SVbox$ and $\Rnext$). It is rather surprising that this does not
happen in the case of the Krom fragment, while the complexity of the
corresponding core fragment jumps from \NLogSpace{} to \NP.

We prove the upper bounds using three different techniques. In Section~\ref{sec:cnf}, we reduce satisfiability in
$\kromLTL\Xu$ to \textsc{2SAT}.  The
existence of models for $\kromLTL\Xallop$-formulas is checked in
Section~\ref{binary} by guessing a small number of types and
exponentially large distances between them (given in binary) and
then using unary automata (and the induced arithmetic progressions) to
verify correctness of the guess in polynomial time. In
Section~\ref{sec:nlogspace}, we design a calculus for
$\coreLTL\Xbox$ in which derivations can be thought of as paths in a 
graph over the propositions labelled by moments of time. Thus, the existence of such 
derivations is essentially the graph reachability problem and can be solved
in \NLogSpace.

%*********************

\section{The Clausal Normal Form for LTL}
\label{sec:cnf}

The \emph{propositional linear-time temporal logic} \PTL{} (see,
e.g.,~\cite{Gabbayetal94,GKWZ03} and references therein) we consider
in this paper is interpreted over the flow of time \mbox{$(\mathbb{Z},<)$}. \PTL-\emph{formulas} are built from propositional variables
$p_0,p_1,\dots$, propositional constants $\top$ and $\bot$, the
Boolean connectives $\land$, $\lor$, $\to$ and $\neg$, and two binary
temporal operators $\Since$ (`since') and $\Until$ (`until'), which are assumed to be `strict.\!'\
So, the other temporal operators mentioned in the introduction can be defined via $\Since$ and $\Until$ as follows:
\begin{align*}
&\Rnext \varphi = \bot \Until \varphi, && \Rdiamond \varphi = \top \Until \varphi, && \Rbox \varphi = \neg \Rdiamond \neg \varphi, && \SVdiamond \varphi = \Ldiamond \Rdiamond \varphi, \\
&\Lnext \varphi = \bot \Since \varphi, && \Ldiamond \varphi = \top \Since \varphi, && \Lbox \varphi = \neg \Ldiamond \neg \varphi, && \SVbox \varphi = \Lbox \Rbox \varphi.
\end{align*}
A \emph{temporal interpretation}, $\M$, defines a truth-relation
between moments of time $n \in \Z$ and propositional variables
$p_i$. We write $\M,n \models p_i$ to indicate that $p_i$ is true at
the moment $n$ in the interpretation $\M$. This truth-relation is
extended to all \PTL-formulas as follows (the Booleans are interpreted as
expected):
\begin{align*}
& \M,n \models \varphi \Until \psi \, \text{ iff } \, \text{there is } k > n \text{ with } \mathfrak{M},k\models \psi \text{ and } \mathfrak{M},m\models\varphi, \text{ for } n < m < k,\\
& \mathfrak{M},n\models \varphi\Since \psi \, \text{ iff } \,  \text{there is } k < n \text{ with } \mathfrak{M},k\models \psi \text{ and } \mathfrak{M},m\models\varphi, \text{ for } k < m < n.
\end{align*}
%
% In particular,
% %
% \begin{align*}
% & \M,n \models \Rnext \varphi \quad \text{iff} \quad \mathfrak{M}, n+1 \models \varphi, \\
% %
% & \mathfrak{M},n\models \Rdiamond \varphi \quad \text{iff}\quad \mathfrak{M},k\models \varphi, \text{ for some $k >n$},\\
% %
% & \mathfrak{M},n\models \Rbox \varphi \quad \text{iff}\quad  \mathfrak{M},k\models \varphi, \text{ for all $k >n$},\\
% %
% & \mathfrak{M},n\models \SVbox \varphi \quad \text{iff}\quad  \mathfrak{M},k\models \varphi, \text{ for all $k \in\Z$},
% \end{align*}
% %
% and similarly for the past-time operators (note that the
% interpretation of $\Until$, $\Rdiamond$ and $\Rbox$ involves only moments of time that are `strictly' in
% the future).
An \PTL-formula $\varphi$ is \emph{satisfiable} if there
is an interpretation $\M$ such that $\M,0 \models \varphi$; in this
case we call $\M$ a \emph{model} of $\varphi$. We denote the length of $\varphi$ by $|\varphi|$. 

Recall that \PTL-formulas of the form \eqref{normal} were said to be in \emph{clausal normal form}, and the class of such formulas was denoted by $\boolLTL\Xallop$. The clauses~\eqref{clause} will often be  represented as $\lambda_1 \land \dots \land \lambda_n \to \lambda_{n+1} \lor \dots \lor \lambda_{n+m}$ (where the empty disjunction is $\bot$ and the empty conjunction is $\top$). 

\begin{lemma}[clausal normal form]\label{lem:clausal:nf}
For every \PTL-formula, one can construct an equisatisfiable $\boolLTL\Xallop$-formula. The construction requires logarithmic space.
\end{lemma}

The proof of this lemma is similar to the proof of~\cite[Theorem~3.3.1]{DBLP:journals/tocl/FisherDP01} and uses fixed-point unfolding and renaming~\cite{DBLP:journals/tocl/FisherDP01,Plaisted86}. For example, we can replace every positive occurrence (that is, an occurrence in the scope of an even number of negations) of $p \Until q$ in a given formula $\varphi$  with a fresh propositional variable $r$ and add the conjuncts $\SVbox(r \to \Rnext q \lor \Rnext p)$, $\SVbox(r \to \Rnext q \lor \Rnext r)$ and $\SVbox(r \to \Rdiamond q)$ to $\varphi$. 
The result contains no positive occurrences of $p \Until q$ and is equisatisfiable with $\varphi$: the first two conjuncts are the fixed-point unfolding  $(p \Until q) \to \Rnext q \lor \bigl(\Rnext p \land \Rnext (p \Until q)\bigr)$, while the last conjunct ensures that the fixed-point is eventually reached. 

The next lemma allows us to consider an even more restricted classes of
formulas. In what follows, we do not distinguish between a set of formulas and the
conjunction of its members, and we write $\SVbox \Phi$ for
the conjunction $\bigwedge_{\chi\in\Phi}\SVbox \chi$.

\begin{lemma}\label{restricted}
Let $\mathcal{L}$ be one of $\PTL_\alpha\Xallop$, $\PTL_\alpha\Xbox$, $\PTL_\alpha\Xu$, for $\alpha \in \{\textit{bool},\textit{horn},\textit{krom},\textit{core}\}$. For any $\mathcal{L}$-formula $\varphi$, one can construct, in log-space, an equisatisfiable $\mathcal{L}$-formula
\begin{equation}\label{eq:kromltl:input}
\Psi\ \land \ \SVbox \Phi,
\end{equation}
where $\Psi$ is a conjunction of propositional variables from $\Phi$, 
and $\Phi$ is a conjunction of clauses of the form~\eqref{clause} containing only
$\Rnext$, $\Lbox$, $\Rbox$ for $\PTL_\alpha\Xallop$, only $\Lbox$, $\Rbox$
for $\PTL_\alpha\Xbox$, and only $\SVbox$ for $\PTL_\alpha\Xu$, in which the temporal operators are not nested.
\end{lemma}
\begin{proof}
First, we take a fresh variable $p$ and replace all the conjuncts of the form $\lambda$ and $\neg\lambda$ in $\varphi$ by $\SVbox(\neg  p \lor \lambda)$ and $\SVbox (\neg p\lor \neg\lambda)$, respectively; set $\Psi = p$. For an $\PTL_\alpha\Xallop$ or $\PTL_\alpha\Xbox$-formula, we replace the temporal literals $\SVbox\lambda$ with $\Rbox\Lbox\lambda$. Then, for each $\Lnext\lambda$, we take a fresh variable, denoted $\overline{\Lnext\lambda}$, replace each occurrence of $\Lnext\lambda$ with $\overline{\Lnext\lambda}$ and add the conjuncts $\SVbox(\Rnext\overline{\Lnext\lambda} \to \lambda)$ and $\SVbox(\lambda \to \Rnext\overline{\Lnext\lambda})$ to the resulting formula. In a similar manner, we use fresh propositional variables as abbreviations for nested temporal operators and obtain the required equisatisfiable formula. 
Clearly, this can be done in logarithmic space.\qed
\end{proof}

%%%%%%%%%%%%%%%%%%
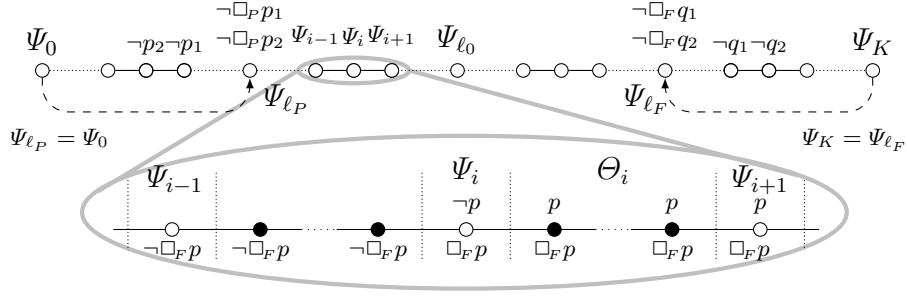
\begin{figure}[t!]%
\centering%
\begin{tikzpicture}[scale=0.92,>=latex,point/.style={circle,draw=black,minimum size=1.8mm,inner sep=0pt}]\footnotesize
\draw[gray!50, ultra thick] (1.5*3-0.8,0) to (-0.2+1,-1.65-0.1);
\draw[gray!50, ultra thick] (1.5*3+0.8,0) to (12.5*0.85-0.8+1.1,-1.42-0.1);
\draw[ultra thick,gray!50] (6*0.85+1,-2.05) ellipse (6.5*0.85 and 1.3*0.85);
\draw[ultra thick,gray!50] (1.5*3,0) ellipse (0.8 and 0.185);
\draw[densely dotted] (0,0) -- (8*1.5,0);
\node[point,fill=white,label=above:{\large $\Psi_0$}] (psi0) at (1.5*0,0) {};
\node[point,fill=white,label=below right:{\large $\Psi_{\ell_\PP}$},label=above:{\begin{tabular}{c}$\neg\Lbox p_1$\\$\neg\Lbox p_2$\end{tabular}}] (psi2) at (1.5*2,0) {};
\node[point,fill=white,label=above:{\large $\Psi_{\ell_0}$}] (psi4) at (1.5*4,0) {};
\node[point,fill=white,label=below left:{\large $\Psi_{\ell_\FF}\hspace*{-0.6em}$},label=above:{\begin{tabular}{c}$\neg\Rbox q_1$\\$\neg\Rbox q_2$\end{tabular}}] (psi6) at (1.5*6,0) {};
\node[point,fill=white,label=above:{\large $\Psi_K$}] (psi8) at (1.5*8,0) {};
\node at (1.5*8-0.25,-1) {$\Psi_K = \Psi_{\ell_\FF}$};
\node at (0.25,-1) {$\Psi_{\ell_\PP} = \Psi_0$};
\foreach \i in {1,3,5,7} {
	\node[point,fill=white]  (1) at (1.5*\i-0.55,0) {};
	\node[point,fill=white]  (2) at (1.5*\i,0) {};
	\node[point,fill=white]  (3) at (1.5*\i+0.55,0) {};
	\draw (1) -- (2);
	\draw (2) -- (3);
}
\begin{scope}[label distance=3pt]
	\node[point,fill=white,label=above:{$\Psi_{i-1}$}]  (1) at (1.5*3-0.55,0) {};
	\node[point,fill=white,label=above:{$\Psi_i$}]  (2) at (1.5*3,0) {};
	\node[point,fill=white,label=above:{$\Psi_{i+1}$}]  (3) at (1.5*3+0.55,0) {};
\end{scope}
\draw[->,dashed,rounded corners=10] (psi0) -- ++(0,-0.6) -| (psi2);  
\draw[->,dashed,rounded corners=10] (psi8) -- ++(0,-0.6) -| (psi6);  
\node[point, label=above:{$\neg p_1$}] at (1.5*1+0.55,0) {};
\node[point, label=above:{$\neg p_2$}] at (1.5*1,0) {};
\node[point, label=above:{$\neg q_1$}] at (1.5*7-0.55,0) {};
\node[point, label=above:{$\neg q_2$}] at (1.5*7,0) {};
\begin{scope}[xshift=6mm,yshift=-23mm,scale=0.85]
\begin{scope}
\clip (6.5,0.25) ellipse (6.5 and 1.3);
\draw (0.5,0) -- (3.7,0);
\draw[dotted] (3.7,0) -- (4.3,0);
\draw (4.3,0) -- (8.7,0);
\draw[dotted] (8.7,0) -- (9.3,0);
\draw (9.3,0) -- (12.5,0);
\node at (1.5,0) [point,fill=white,label=below:{$\neg\Rbox p$}] {};
\node at (3,0) [point,fill=black,label=below:{$\neg\Rbox p$}] {};
%\node at (4,0) [label=above:{$\ldots$}] {};
\node at (5,0) [point,fill=black,label=below:{$\neg\Rbox p$}] {};
\node at (6.5,0) [point,fill=white,label=below:{$\Rbox p$},label=above:{$\neg p$}] {};
\node at (8,0) [point,fill=black,label=below:{$\Rbox p$},label=above:{$p$}] {};
%\node at (9,0) [label=above:{$\ldots$}] {};
\node at (10,0) [point,fill=black,label=below:{$\Rbox p$},label=above:{$p$}] {};
\node at (11.5,0) [point,fill=white,label=below:{$\Rbox p$\hspace*{3mm}},label=above:{$p$}] {};
\draw[densely dotted] (0.75,-0.5) -- ++(0,1.5);
\draw[densely dotted] (2.25,-0.5) -- ++(0,1.5);
\draw[densely dotted] (5.75,-0.5) -- ++(0,1.5);
\draw[densely dotted] (7.25,-0.5) -- ++(0,1.5);
\draw[densely dotted] (10.75,-0.5) -- ++(0,1.5);
\draw[densely dotted] (12.25,-0.5) -- ++(0,1.5);
\end{scope}
\node at (1.55,0.85) {\large $\Psi_{i-1}$};
\node at (6.5,1) {\large $\Psi_i$};
\node at (9,1) {\large $\Theta_i$};
\node at (11.5,0.85) {\large $\Psi_{i+1}$};
\end{scope}
\end{tikzpicture}
\caption{The structure of a model in Lemma~\ref{l:structure}.}\label{fig:structure}
\end{figure}
%%%%%%%%%%%%%% 

We now characterise the structure of interpretations satisfying
formulas $\varphi^*$ of the form~\eqref{eq:kromltl:input} in a way
similar to other known descriptions of temporal models; see,
e.g.,~\cite{Gabbayetal94,GKWZ03}. This characterisation will be used in the upper bound proofs of Theorems~\ref{lem:bin-ltl:krom-diamond-next-np} and~\ref{thm:hornLTL}.
For each $\Rbox p$ in $\Phi$, we take a fresh propositional variable, $\overline{\Rbox p}$, and call it the \emph{surrogate} of $\Rbox p$; likewise, for each $\Lbox p$ in $\Phi$ we take its surrogate 
$\overline{\Lbox p}$. Let $\overline{\Phi}$ be the result of replacing all the $\Box$-literals in $\Phi$ with their surrogates. 
%It should be clear that $\varphi$ is equisatisfiable with
%%
%\begin{equation*}
%\overline{\varphi} \ \ \ =  \ \ \ \Psi\land\SVbox\overline{\Phi}\ \ \ \ \land \ \bigwedge_{\Rbox L \text{ occurs in } \Phi} \hspace*{-1.5em}\SVbox(\overline{\Rbox L} \leftrightarrow \Rbox L) \ \ \ \ \land\ \bigwedge_{\Lbox L \text{ occurs in } \Phi} \hspace*{-1.5em}\SVbox(\overline{\Lbox L} \leftrightarrow \Lbox L).
%\end{equation*}
%
By a \emph{type} for $\overline{\Phi}$ we mean any set of literals
that contains either $p$ or $\neg p$ (but not both), for each variable $p$ in
$\overline{\Phi}$ (including the surrogates). 

The proof of the following lemma is standard; cf.~\cite{Gabbayetal94,GKWZ03}. The reader may find
useful Fig.~\ref{fig:structure} illustrating the conditions of the
lemma.

\begin{lemma}[structure of models]\label{l:structure}
Let $\varphi$ be an $\boolLTL\Xallop$-formula of the form~\eqref{eq:kromltl:input} and $K = |\varphi| + 4$. Then $\varphi$ is satisfiable iff there exist integers $m_0 < m_1 < \dots < m_K$
%\textup{(}where $k$ does not exceed the number of $\Box$-subformulas of $\varphi^*$\textup{)} 
and types $\Psi_0, \Psi_1, \dots, \Psi_K$ for $\overline{\Phi}$ such that\textup{:}
\begin{description}\itemsep=4pt
\item[\textup{(B$_0$)}] $m_{i+1} - m_i < 2^{|\overline{\Phi}|}$, for $0 \leq i < K$\textup{;}
%\item for each $0 \leq i \leq k+2$ and each variable $p$ in $\varphi$ (including surrogates), either $p \in
%  \Psi_i$ or $\neg p\in \Psi_i$;
%
\item[\textup{(B$_1$)}] there exists $\ell_0$, $0 < \ell_0 < K$,  
such that $\Psi \subseteq \Psi_{\ell_0}$\textup{;}
\item[\textup{(B$_2$)}] 
$\overline{\Rbox p}\in \Psi_i 
\Rightarrow p,\overline{\Rbox p}\in \Psi_{i+1}$ \ \ and \ \ $\overline{\Rbox p} \in \Psi_{i+1} \setminus \Psi_i \Rightarrow p\notin \Psi_{i+1}$ $(0 \leq i < K)$,\\
$\overline{\Lbox p}\in \Psi_i
\Rightarrow p,\overline{\Lbox p}\in \Psi_{i-1}$ \ \ and \ \  
$\overline{\Lbox p} \in \Psi_{i-1} \setminus
  \Psi_i  
\Rightarrow p\notin \Psi_{i-1}$ $(0 < i \leq K)$\textup{;} 
\item[\textup{(B$_3$)}]  there exist $\ell_\FF < K$ and $\ell_\PP > 0$ such that  %, for each $\Rbox p$ in $\Phi$,
\begin{itemize}\itemsep=2pt
\item $\Psi_{\ell_\FF}=\Psi_K$ and, for each $\neg\overline{\Rbox p} \in \Psi_{\ell_\FF}$, there is $j \geq \ell_\FF$ with $\neg p \in \Psi_j$,
\item $\Psi_{\ell_\PP}=\Psi_0$ and, for each $\neg\overline{\Lbox p} \in \Psi_{\ell_\PP}$, there is $j \leq \ell_\PP$ with $\neg p \in \Psi_j$\textup{;}
\end{itemize}
\item[\textup{(B$_4$)}]
the following formulas are consistent, for $0 \leq i < K$\textup{:} 
\begin{equation*}
\psi_i ~=~  \Psi_i \ \ \ \land \bigwedge_{k = 1}^{m_{i+1} - m_i - 1} \hspace*{-1.5em}\Rnext^k 
  \Theta_i \ \ \ \ \land \ \ \ \Rnext^{m_{i+1} - m_i} \Psi_{i+1} \ \ \ \ \land \ \ \ \ \SVbox \overline{\Phi},
\end{equation*}
where $\Rnext^k \Psi$ is the result of attaching $k$ operators $\Rnext$ to each literal in $\Psi$ and
\begin{multline*}
\Theta_i \ \ = \ \
 \bigl\{ p, \ \overline{\Rbox p} \mid \overline{\Rbox p}\in \Psi_i \bigr\} \cup
\bigl\{ \neg\overline{\Rbox p} \mid \neg \overline{\Rbox p}\in \Psi_i \bigr\} \cup {} \\
 \bigl\{ p, \ \overline{\Lbox p} \mid \overline{\Lbox p}\in\Psi_{i+1} \bigr\} \cup
\bigl\{ \neg\overline{\Lbox p} \mid \neg\overline{\Lbox p}\in \Psi_{i+1} \bigr\}.
\end{multline*}
\end{description}
\end{lemma}
\begin{proof}
  ($\Rightarrow$) Let $\mathfrak{M},0\models\varphi$. Denote by
  $\Psi(m)$ the type for $\overline\varphi$ containing all literals that
  hold at the moment $m$ in $\mathfrak{M}$. As the number of types is
  finite, there is $m_{\FF} > 0$ such that each type in the sequence
  $\Psi(m_{\FF}),\Psi(m_{\FF}+1),\dots$ appears infinitely
  often; similarly, there is $m_{\PP} < 0$ such that each type in the
  sequence $\Psi(m_{\PP}),\Psi(m_{\PP}-1),\dots$ appears
  infinitely often.  Then, for each literal $\Rbox L$ of $\Phi$,
  we have one of three options: (1) $L$ is always true in $\mathfrak{M}$, in which case we set
  $m_{\Box_F L} = 0$; (2)  
  there is $m_{\Box_F L}$ such that $\mathfrak{M},m_{\Box_F
    L}\models \neg L \land \Rbox L$, in which case $m_{\PP} < m_{\Box_F L}
  < m_{\FF}$; or (3) $\Rbox L$ is always false in $\mathfrak{M}$, in which
  case $L$ is false infinitely often after the moment $m_{\FF}$, and so
  there is $m_{\Box_F L}\geq m_{\FF}$ such that $\mathfrak{M},m_{\Box_F
    L}\models \neg L$. Symmetrically, for each literal $\Lbox L$ of
  $\Phi$, we have one of three options: (1) $L$ is always true in $\mathfrak{M}$, in which case we set
  $m_{\Box_P L} = 0$; (2) there is an $m_{\Box_P L}$ such that $m_{\PP} <
  m_{\Box_P L} < m_{\FF}$ and $\mathfrak{M},m_{\Box_P L}\models \neg L
  \land \Lbox L$; or (3) $\Lbox L$ is always false in $\mathfrak{M}$, in
  which case there is $m_{\Box_P L} \leq m_{\PP}$ such that
  $\mathfrak{M},m_{\Box_P L}\models \neg L$. Let
  $m_1<m_2<\dots<m_{K-1}$ be an enumeration of the set (padded if necessary)
\begin{equation*}
M = \{0,m_{\PP},m_{\FF}\}\cup\{ m_{\Box_F L} \mid \Rbox L \text{ occurs in } \Phi \}\cup\{ m_{\Box_P L} \mid \Lbox L \text{ occurs in } \Phi\}.
\end{equation*}
Let $m_K > m_{K-1}$ be such that $\Psi(m_K) =
\Psi(m_{\FF})$ and let $m_0 < m_1$ be such that $\Psi(m_0) =
\Psi(m_{\PP})$. We then set $\Psi_i = \Psi(m_i)$, for $0
\leq i \leq K$.  Let $\ell_0$, $\ell_P$ and $\ell_F$ be such that
$m_{\ell_0} = 0$, $m_{\ell_P} = m_{\PP}$ and $m_{\ell_F} = m_{\FF}$. It should
be clear that~\textbf{(B$_1$)}--\textbf{(B$_4$)} hold. Finally, given
a model of $\overline\varphi$ with two moments $m$ and $n$ such that the types
at $m$ and $n$ coincide, we can construct a new model for $\varphi$ by
`removing' the states $i$ with $m \leq i <
n$. Since the number of distinct types is bounded by $2^{|\overline{\Phi}|}$,
by repeated applications of this construction we can further
ensure~\textbf{(B$_0$)}.

\smallskip

($\Leftarrow$) We construct a model $\mathfrak{M}$ of $\varphi$
by taking finite cuts of the models $\mathfrak{M}_i$ of the formulas
in~\textbf{(B$_4$)}: between the moments $m_0$ and $m_K$, the
model $\mathfrak{M}$ coincides with the models
$\mathfrak{M}_0,\dots,\mathfrak{M}_{K-1}$ so that at the moment $m_i$
in $\mathfrak{M}$ we align the moment 0 of $\mathfrak{M}_i$, and at the
moment $m_{i+1}$ we align the moment $m_{i+1}-m_i$ of
$\mathfrak{M}_i$, which coincides with the moment 0 of
$\mathfrak{M}_{i+1}$ because both are defined by $\Psi_{i+1}$; before
the moment $m_0$, the model $\mathfrak{M}$ repeats infinitely often  its own
fragment between $m_0$ and $m_{\ell_P}$, and after $m_K$ it
repeats infinitely often its fragment between $m_{\ell_F}$ and $m_K$
(both fragments contain more than one
state). It is readily seen that
$\mathfrak{M},m_{\ell_0}\models\varphi$.\qed
\end{proof}

The intuition
behind this lemma is as follows (see Fig.~\ref{fig:structure}). If
$\varphi$ is satisfiable, then it has a model $\M$ that consists of the 
initial fragments of models $\M_i$ of the formulas $\psi_i$: namely, the types of the moments $m_i,\dots,m_{i+1}$ in $\M$ coincide with the types of the moments $0,\dots, (m_{i+1}-m_i)$ in $\M_i$. By~\iref{(B$_4$)}, we  have $\M,0 \models \SVbox \overline{\Phi}$. Then \iref{(B$_1$)} makes sure  that $\M,0 \models\Psi$. Conditions \iref{(B$_2$)} and \iref{(B$_3$)} guarantee that if $\overline{\Rbox p} \in \Psi_i$ then $p \in \Psi_j$ for all types $\Psi_j$ located to the right of $\Psi_i$ in Fig.~\ref{fig:structure} and, conversely, if $\overline{\Rbox p} \notin \Psi_i$ then $\neg p \in \Psi_j$, for some $\Psi_j$ to the right of $\Psi_i$; and symmetrically for the $\Lbox$-literals. It follows  that $\M,0\models \SVbox \Phi$.

\begin{theorem}\label{thm:kromLTL}
The satisfiability problem for $\kromLTL\Xu$-formulas is in \NLogSpace.
\end{theorem}
\begin{proof}
The proof is by reduction to 2SAT. Let $\varphi = \Psi\land\SVbox \Phi$ be of the form~\eqref{eq:kromltl:input} and let $p_1, \dots, p_N$ be the variables of $\varphi$. For each variable $p_i$ of $\varphi$, we take $N+1$ variables $p_i^m$, for $0 \leq m \leq N$. We take a special fresh variable $\overline{\SVbox p_i}$  for each $\SVbox p_i$ and extend the $\cdot^m$ notation by taking $(\SVbox p_i)^m = \overline{\SVbox p}$.
Then the literals of the form $p_i$ and $\neg p_i$ in $\Psi$ give rise to clauses $p^0_i$  and $\neg p^0_i$, respectively, each $\SVbox (\lambda_1\lor \lambda_2)$ in $\Phi$ gives rise to $N + 1$ clauses  $\lambda_1^m\lor \lambda_2^m$, for $0 \leq m \leq N$, and  similarly for other forms of clauses in $\Phi$. Finally, we add the clauses $\overline{\SVbox p_i} \to p_i^m$, for all $0 \leq m \leq N$, and $\neg \overline{\SVbox p_i} \to \neg p^i_i$, which express the semantics of $\SVbox p_i$ (without loss of generality we may assume that $p_i$ is false at moment $i$ in case $\SVbox p_i$ is false). Clearly, $\varphi$ is satisfiable iff the above set of binary clauses is satisfiable.\qed
\end{proof}

%*********************

\section{Binary-Clause LTL and Arithmetic Progressions}\label{binary}

In this section, we prove \NP-completeness of the satisfiability
problem for $\kromLTL\Xallop$ and $\coreLTL\Xallop$. The key
ingredient of the proof of the upper bound is an encoding of condition
\iref{(B$_4$)} for \emph{binary clauses} by means of arithmetic
progressions (via unary automata). The proof of the lower bound is by
reduction of the problem whether a given set of
arithmetic progressions covers all the natural numbers.

Let $\varphi$ be an $\kromLTL\Xallop$-formula of the
form~\eqref{eq:kromltl:input}.  By Lemma~\ref{l:structure}, to check
satisfiability of $\varphi$  
it suffices to guess $K + 1$ types for $\overline{\Phi}$ and $K$ natural
numbers $n_i = m_{i+1} - m_i$, for $0 \leq i < K$, whose binary
representation, by~\iref{(B$_0$)}, is polynomial in
$|\overline{\Phi}|$. Evidently, \iref{(B$_1$)}--\iref{(B$_3$)} can be
checked in polynomial time. Our aim now is to show that~\iref{(B$_4$)}
can also be verified in polynomial time, which will give a nondeterministic polynomial-time algorithm for checking satisfiability of $\kromLTL\Xallop$-formulas. 
\begin{theorem}\label{lem:bin-ltl:krom-diamond-next-np}
The satisfiability problem for $\kromLTL\Xallop$-formulas is in \NP.
\end{theorem}
\begin{proof}
  In view of
  Lemma~\ref{restricted}, we write $\nxt$ in place of $\Rnext$. We denote propositional literals ($p$ or $\neg p$) by $L$ and
  temporal literals ($p$, $\neg p$, $\nxt p$ or $\neg\nxt p$) by $D$.
We assume that $\nxt\neg p$ is the same as $\neg \nxt p$. 
%$\neg \neg p = p$ and $\nxt \neg L = \neg \nxt L$. 
We
  use $\psi_1 \models \psi_2$ as a shorthand for `$\mathfrak M,0
  \models \psi_2$ whenever $\mathfrak M,0 \models \psi_1$, for any
  interpretation $\mathfrak M$.'
Thus, the problem is as follows: given a set $\Phi$ of binary clauses of the form $D_1 \lor D_2$, types $\Psi$ and $\Psi'$ for $\Phi$, a set $\Theta$ 
of propositional literals and a number $n > 0$ (in binary), decide whether 
\begin{equation}\label{eq:reachability}
\Psi  \ \ \land \ \ \bigwedge\nolimits_{k = 1}^{n-1} \nxt^k \Theta \ \ \land \ \ \nxt^n \Psi' \ \ \land \ \  \SVbox\Phi
\end{equation}
has a satisfying interpretation. For $0 \leq k \leq n$, we set: 
\begin{align*}
F^k_{\Phi}(\Psi) & = \bigl\{ L' \mid L\land \SVbox \Phi \models \nxt^k L', \text{ for } L\in\Psi \bigr\},\\
P^k_{\Phi}(\Psi') & = \bigl\{ L \mid \nxt^k L'\land \SVbox \Phi \models L, \text{ for } L'\in\Psi' \bigr\}.
\end{align*}

\begin{lemma}
Formula~\eqref{eq:reachability} is satisfiable iff the following conditions hold\textup{:}
\begin{description}
\item[\textup{(L$_1$)}] 
$F^0_{\Phi}(\Psi)\subseteq\Psi$, $F^n_{\Phi}(\Psi)\subseteq\Psi'$ and $P^0_{\Phi}(\Psi')\subseteq\Psi'$, $P^n_{\Phi}(\Psi')\subseteq\Psi$\textup{;}
\item[\textup{(L$_2$)}] 
 $\neg L\notin F^k_{\Phi}(\Psi)$ and $\neg L\notin P^{n-k}_{\Phi}(\Psi')$,  for all $L\in\Theta$ and $0 < k < n$.
\end{description}
\end{lemma}
\begin{proof}
Clearly, if~\eqref{eq:reachability} is satisfiable
then the above conditions hold.  For the converse direction,  observe that if $L'\in F^k_\Phi(\Psi)$ then, since $\Phi$ is a set of binary clauses, there
  is a sequence of $\nxt$-prefixed literals $\nxt^{k_0} L_0 \leadsto \nxt^{k_1} L_1 \leadsto \dots \leadsto
  \nxt^{k_m} L_m$ such that $k_0 = 0$, $L_0\in\Psi$, $k_m = k$, $L_m =
  L'$, each $k_i$ is between $0$ and $n$ and the $\leadsto$ relation is defined by taking $\nxt^{k_i} L_i \leadsto \nxt^{k_{i+1}} L_{i+1}$ just in one of the three cases: $k_{i+1} = k_i$ and $L_i \to L_{i+1}\in\Phi$ or $k_{i+1} = k_i + 1$ and $L_i \to \nxt L_{i+1}\in\Phi$ or $k_{i+1} = k_i - 1$ and $\nxt L_i \to L_{i+1}\in\Phi$ (we assume that, for example, $\neg q \to \neg p\in\Phi$ whenever $\Phi$ contains $p \to q$).  
  So, suppose conditions~\textbf{(L$_1$)}--\textbf{(L$_2$)} hold. We construct an
  interpretation
  satisfying~\eqref{eq:reachability}. By~\textbf{(L$_1$)}, both
  $\Psi\land\SVbox\Phi$ and $\nxt^n \Psi'\land\SVbox\Phi$ are
  consistent. So, let $\mathfrak{M}_\Psi$ and $\mathfrak{M}_{\Psi'}$ be such that
  $\mathfrak{M}_\Psi,0\models\Psi\land\SVbox \Psi$ and $\mathfrak{M}_\Psi,n\models\Psi'\land\SVbox \Psi$, respectively.
  Let $\mathfrak{M}$ be an interpretation that coincides with $\mathfrak{M}_\Psi$ for all
  moments $k \leq 0$ and with $\mathfrak{M}_{\Psi'}$ for all $k \geq n$;
  for the remaining $k$, $0 < k < n$, it is defined as follows. First,
  for each $p\in \Theta$ , we make $p$ true at $k$ and, for each $\neg
  p\in\Theta$, we make $p$ false at $k$; such an assignment exists due
  to \textbf{(L$_2$)}. Second, we extend the assignment by making $L$
  true at $k$ if $L\in F^k_{\Phi}(\Psi)\cup
  P^{n-k}_{\Phi}(\Psi')$. Observe that we have $\{p,\neg p\}\nsubseteq
  F^k_{\Phi}(\Psi)\cup P^{n-k}_{\Phi}(\Psi')$: for otherwise $L\land
  \SVbox\Phi\models \nxt^k p$ and $\nxt^{n-k} L' \land
  \SVbox\Phi\models \neg p$, for some $L\in\Psi$ and $L'\in\Psi'$, whence $L\land \SVbox\Phi\models \nxt^n
  \neg L'$, contrary to~\textbf{(L$_1$)}. Also, by~\textbf{(L$_2$)}, any
  assignment extension at this stage does not contradict the choices
  made due to $\Theta$.  Finally, all propositional variables not
  covered in the previous two cases get their values from $\mathfrak{M}_\Psi$
  (or $\mathfrak{M}_{\Psi'}$). We note that the last choice does not depend on the
  assignment that is fixed by taking account of the consequences of
  $\SVbox\Phi$ with $\Psi$, $\Psi'$ and $\Theta$ (because if the value
  of a variable depended on those sets of literals, the respective
  literal would be among the logical consequences and would have been
  fixed before).\qed
\end{proof}

Thus, it suffices to show that conditions
\iref{(L$_1$)} and \iref{(L$_2$)} can be checked in polynomial time.
First, we claim that there is a polynomial-time algorithm which, given
a set $\Phi$ of binary clauses of the form $D_1\lor D_2$, %(as in~\eqref{eq:reachability}),
constructs a set $\Phi^*$ of binary clauses that is `sound and
complete' in the following sense: 
\begin{description}
\item[(S$_1$)] $\SVbox\Phi^*\models \SVbox\Phi$;  
%\item[(S$_1$)] if $\SVbox\Phi \models \SVbox (L_1 \lor L_2)$ then $L_1
%  \lor L_2\in\Phi^*$, for all literals $L_1$ and $L_2$;
\item[(S$_2$)] if $ \SVbox\Phi\models \SVbox (L \to \nxt^k L_k)$ then
  either $k = 0$ and $L \to L_0\in\Phi^*$, or $k \geq 1$ and there are $L_0,L_1,\dots,L_{k-1}$  with $L = L_0$ and 
  $L_i\to\nxt L_{i+1}\in\Phi^*$, for $0 \leq i < k$.
\end{description}
Intuitively, the set $\Phi^*$ makes explicit the
consequences of $\SVbox \Phi$ and can be constructed in time $(2|\Phi|)^2$
%,
%where $k$ is the number of
%iterals built from the variables in $\Phi$ 
(the number of temporal literals in $\Phi^*$ is bounded by the doubled length $|\Phi|$ of $\Phi$ as each of its literal  can only be prefixed by $\nxt$). Indeed, we
start from $\Phi$ and, at each step, add $D_1\lor D_2$ to $\Phi$ if
it contains both $D_1 \lor D$ and $\neg D\lor D_2$; we also add $L_1
\lor L_2$ if $\Phi$ contains $\nxt L_1 \lor \nxt L_2$ (and \emph{vice
versa}). This procedure is sound since we only add consequences of
$\SVbox\Phi$; completeness follows from the completeness proof for
temporal resolution~\cite[Section~6.3]{DBLP:journals/tocl/FisherDP01}.

Our next step is to encode $\Phi^*$ by means of unary automata.
Let $L$, $L'$ be literals. Consider a  nondeterministic finite automaton $\mathfrak{A}_{L,L'}$ over $\{0\}$  such that the literals of~$\Phi^*$ are its states, with $L$ being the initial state and $L'$ the only accepting state, and $\bigl\{(L_1,L_2) \mid L_1 \to \nxt L_2\in\Phi^* \bigr\}$ is its transition relation. By~\iref{(S$_1$)} and~\iref{(S$_2$)},  for all $k > 0$, we have
\begin{align*}
\mathfrak{A}_{L,L'} \text{ accepts } 0^k \quad\text{iff}\quad \SVbox\Phi\models \SVbox(L \to \nxt^k L'). 
\end{align*}
Then both $F^k_\Phi(\Psi)$ and $P^k_\Phi(\Psi')$ can be defined in terms of the
language of $\mathfrak{A}_{L,L'}$:
\begin{align*}
F_\Phi^k(\Psi) & = \bigl\{ L' \mid \mathfrak{A}_{L,L'} \text{ accepts } 0^k, \text{ for } L\in\Psi\bigr\},\\[-2pt]
P_\Phi^k(\Psi') & = \bigl\{ L \mid \mathfrak{A}_{\neg L,\neg L'} \text{ accepts } 0^k, \text{ for } L'\in\Psi'\bigr\}
\end{align*}
(recall that $\nxt^k L' \to L$ is equivalent to $\neg L \to \nxt^k \neg L'$). Note that the numbers $n$ and $k$ in conditions~\iref{(L$_1$)} and~\iref{(L$_2$)} are in general exponential in the length of $\Phi$ and, therefore, the automata $\mathfrak{A}_{L,L'}$ do not immediately provide a polynomial-time procedure for checking these conditions: although it can be shown that if~\iref{(L$_2$)}  does not hold then it fails for a polynomial number $k$, this is not the case for \iref{(L$_1$)}, which requires the accepting state to be reached in a fixed (exponential) number of transitions. Instead, we use the~\emph{Chrobak normal form}~\cite{chrobak-ufa} to decompose the automata into a polynomial number of polynomial-sized arithmetic progressions (which can have an exponential common period; cf.~the proof of Theorem~\ref{lem:bin-ltl:core-diamond}). In what follows,  given $a$ and $b$, we denote by $a + b\mathbb{N}$ the set $\{ a + b m \mid m\in\mathbb{N} \}$ (the arithmetic progression with initial term $a$ and common difference $b$).

It is known that every $N$-state unary automaton
$\mathfrak{A}$ can be converted (in polynomial time) into an equivalent automaton in Chrobak normal form (e.g., by using Martinez's algorithm~\cite{to-ufa}), which has $O(N^2)$ states and gives rise to
$M$ arithmetic progressions $a_1 + b_1\mathbb{N},\dots,a_M + b_M\mathbb{N}$ such that
\begin{description}
\item[(A$_1$)] $M\leq O(N^2)$ and $0 \leq a_i, b_i \leq N$, for $1 \leq i \leq
  M$;
\item[(A$_2$)] $\mathfrak{A}$ accepts $0^k$ iff $k \in a_i + b_i\mathbb{N}$,
for some $1 \leq i \leq M$. 
\end{description}
%
%As the states of $\mathfrak{A}_{L,L'}$ are literals of $\Phi$, 
By construction, %of $\mathfrak{A}_{L,L'}$, 
the
number of arithmetic progressions is bounded by a quadratic function
in the length of $\Phi$.

We are now in a position to give a
polynomial-time algorithm for checking~\iref{(L$_1$)} and~\iref{(L$_2$)}, 
% consistency of~\eqref{eq:reachability}.
%By~\iref{(S$_1$)}, condition~\iref{(L$_1$)} can trivially be
%checked in quadratic time once the set $\Phi^*$ is
%constructed. 
%Conditions~\iref{(L$_1$)}\nb{to be expanded} and \iref{(L$_2$)} 
which requires solving
Diophantine equations.  In~\iref{(L$_2$)}, for example, to verify that, for each $p\in\Theta$, we have $\neg p\notin F_\Phi^k(\Psi)$, for all $0 < k < n$, we take the automata $\mathfrak{A}_{L,\neg p}$, for $L\in\Psi$, and
transform them into the Chrobak normal form to obtain arithmetic
progressions $a_i+ b_i\mathbb{N}$, for $1 \leq i \leq M$. Then there is $k$,
$0 < k < n$, with $\neg p\in F_\Phi^k(\Psi)$ iff one of the equations
$a_i+b_i m=k$ has an integer solution, for some $k$, $0 < k < n$. The latter can
be verified by taking the integer $m = \lfloor -a_i/b_i\rfloor$
and checking whether either $a_i + b_im$ or $a_i + b_i(m + 1)$
belongs to the open interval $(0,n)$, which can clearly be done in
polynomial time.

This completes the proof of Theorem~\ref{lem:bin-ltl:krom-diamond-next-np}.
\qed
\end{proof}

The matching lower bound for $\coreLTL\Xallop$-formulas, even without $\Rbox/\Lbox$, can be obtained using \NP{}-hardness of deciding inequality of regular languages  over a unary alphabet~\cite{StockmeyerM73}. In the proof of Theorem~\ref{lem:bin-ltl:core-diamond}, we give a  more direct reduction of the \NP-complete problem 3SAT and repeat the argument of \cite[Theorem 6.1]{StockmeyerM73} to construct a small number of arithmetic progressions (each with a small initial term and common difference) that give rise to models of exponential size.

\begin{theorem}\label{lem:bin-ltl:core-diamond}
  The satisfiability problem for $\coreLTL\Xallop$-formulas is \NP-hard.
\end{theorem}
\begin{proof} 
The proof is by reduction of 3SAT. 
Let $f = \bigwedge_{i = 1}^n C_i$ be a 3CNF with variables $p_1,\dots,p_m$ and clauses $C_1,\dots,C_n$. By a propositional assignment for $f$ we understand a function $\sigma \colon \{p_1, \dots, p_m\} \to \{0,1\}$. We  represent such assignments by sets of positive natural numbers. More precisely, 
let $P_1,\dots,P_m$ be the first $m$ prime numbers; it is known that $P_m$ does not exceed $O(m^2)$~\cite{Apostol76}.  A natural number $k>0$ is said to \emph{represent} an assignment $\sigma$ if 
$k$ is equivalent to $\sigma(p_i)$ modulo $P_i$, for all $i$, $1 \leq i \leq m$.  
Clearly, not every natural number represents an assignment since each element of
\begin{equation}\label{eq:NP:AP-not}
j + P_i\cdot\mathbb{N},\qquad\text{ for } 1 \leq i \leq m \text{ and } 2 \leq j < P_i,
\end{equation}
is equivalent to $j$ modulo $P_i$ with $j \geq 2$. On the other hand, every natural number that does not represent an assignment belongs to one of those arithmetic progressions  (see Fig.~\ref{fig:stockmeyer}).

\begin{figure}[t]
\centering{%
\tabcolsep=2pt\scriptsize%
\newcolumntype{g}{>{\columncolor{gray!30}}c}
\begin{tabular}{c|g|cccc|g|ccc|g|cccc|g|g|cccc|g|ccc|g|cccc|g}
& \bf 1 & 2 & 3 & 4 &  5 & \bf 6 & 7 & 8 & 9 & \bf 10 & 11 & 12 & 13 & 14 & \bf 15 & \bf 16 & 17 & 18 & 19 &  20 & \bf 21 & 22 & 23 & 24 & \bf 25 & 26 & 27 & 28 & 29 & \bf 30 \\\hline
2 & 1 & 0 & 1 & 0 & 1 & 0 & 1 & 0 & 1 & 0 & 1 & 0 & 1 & 0 & 1 & 0 & 1 & 0 & 1& 0 & 1 & 0 & 1 & 0 & 1 & 0 & 1 & 0 & 1 & 0\\
3 & 1 &  & 0 & 1 &  & 0 & 1 &  & 0 & 1 &  & 0 & 1 &  & 0 & 1 &  & 0 & 1 &  & 0 & 1 &  & 0 & 1 &  & 0 & 1 &  & 0\\
5 & 1 &  &  &  &  0 & 1 &  &  &  &  0 & 1 &  &  &  & 0 & 1 &  &  &  & 0 & 1 &  &  &  & 0 & 1 &  &  &  & 0\\
\end{tabular}
}
\caption{Positive numbers encoding assignments for 3 variables $p_1,p_2,p_3$ (shaded).}\label{fig:stockmeyer}
\end{figure}

Let $C_i$ be a clause in $f$, say, $C_i = p_{i_1} \lor \neg p_{i_2}  \lor p_{i_3}$. Consider 
\begin{equation}\label{eq:NP:AP-clause}
P_{i_1}^{1} P_{i_2}^{0} P_{i_3}^{1} + P_{i_1}P_{i_2}P_{i_3}\cdot \mathbb{N}.
\end{equation}
A natural number represents an assignment that makes $C_i$ true iff it does not belong to the progressions~\eqref{eq:NP:AP-not} and \eqref{eq:NP:AP-clause}.  In the same way we construct a progression of the form~\eqref{eq:NP:AP-clause} for every clause in $f$. 
Thus, a natural number $k > 0$ \emph{does not} belong to the constructed progressions of the form~\eqref{eq:NP:AP-not} and~\eqref{eq:NP:AP-clause} iff $k$ represents a satisfying assignment for $f$.
 
To complete the proof, we show that the defined progressions can be encoded in $\coreLTL\Xallop$. Take a propositional variable $d$ (it will be shared by  all formulas below).
Given an arithmetic progression $a+b\mathbb{N}$ (with $a\geq 0$ and $b > 0$), let
\begin{multline*}
\theta_{a,b} \ \ = \ \ u_0  \land\bigwedge\nolimits_{j = 1}^a \SVbox (u_{j-1} \to \Rnext u_j) \land{}\\[-4pt] \SVbox (u_a \to v_0) \land \bigwedge\nolimits_{j = 1}^b \SVbox (v_{j-1} \to \Rnext v_j) \land \SVbox(v_b \to v_0) \land \SVbox(v_0 \to d),
\end{multline*}
where $u_0,\dots,u_a$ and $v_0,\dots,v_b$ are fresh propositional variables.
It is not hard to see that,  in every model of $\theta_{a,b}$, if $k$ belongs to $a + b\mathbb{N}$, then $d$ is true at moment $k$. 
Thus, we take a conjunction $\varphi_f$ of the $\theta_{a,b}$ for  arithmetic progressions~\eqref{eq:NP:AP-not} and~\eqref{eq:NP:AP-clause} together with $p \land \SVbox(\Rnext p \to p) \land \SVbox (p \to d) \land \SVbox (\neg \SVbox d)$,
where $p$ is a fresh variable (the last formula makes both $p$ and $d$ true at all moments $k \le 0$). The size of the
$\coreLTL\Xallop$-formula  $\varphi_f$ is $O(n\cdot m^6)$.  It is readily checked
that $\varphi_f$ is satisfiable iff $f$ is satisfiable.
\qed
\end{proof}

%***************

\section{Core and Horn Fragments without Next-Time}

Let $\varphi$ be an $\hornLTL\Xbox$-formula.  By applying Lemma~\ref{restricted},  we  can transform $\varphi$ to the form $\Psi\land \SVbox \Phi^+ \land \SVbox \Phi^-$, where $\Psi$ is a set of propositional variables while $\Phi^+$ and $\Phi^-$ are sets of \emph{positive} and \emph{negative clauses} of the form
\begin{equation}\label{eq:horn-clauses}
\lambda_1 \land \lambda_2 \land 
\dots \land \lambda_{k-1} \to \lambda_k \quad\text{ and }\quad \neg \lambda_1 \lor \neg \lambda_2 \lor \dots \lor \neg \lambda_k, 
\end{equation}
respectively.
Trivially, $\Psi\land\SVbox\Phi^+$ is satisfiable. Since all clauses in $\Phi^+$ have at most one positive literal and are constructed from variables possibly prefixed by $\Rbox$ or $\Lbox$, the formula $\Psi\land\SVbox\Phi^+$ has a \emph{canonical model} $\K$ defined by taking
\begin{equation*}
\K,n\models p \quad\text{ iff }\quad \M,n\models p,\ \ \text{ for every model } \M \text{ of } \Psi\land\SVbox\Phi^+, \ n \in\mathbb{Z}
\end{equation*}
(indeed, $\K,0\models \Psi\land\SVbox\Phi^+$ follows from the observation that  $\K,n\models \Rbox p$ iff $\M,n\models \Rbox p$,  for every model $\M$ of $\Psi\land\SVbox\Phi^+$; and similarly for $\Lbox p$). 
If we consider the canonical model $\K$ in the context of
Lemma~\ref{l:structure} then, since the language
does not contain $\Rnext$ or $\Lnext$, we have $m_{i+1} -
m_i = 1$ for all $i$. Thus, $\K$ can be thought
of as a sequence of $(\ell_\FF - \ell_\PP +1)$-many states, the first and last of which repeat
indefinitely. Let $K = |\varphi| + 4$. 

\begin{figure}[t]
\centering
\begin{tikzpicture}[xscale=0.75,yscale=0.7,>=latex,point/.style={circle,draw=black,minimum size=1.8mm,inner sep=0pt}]
\footnotesize%
\node[point,label=above:{\scriptsize $-K$}] (a1) at (0,0) {};
\node[point] (a2) at (1.5,0) {};
\node[point] (a3) at (3,0) {};
\node[point] (a4) at (4.5,0) {};
\node[point] (a5) at (6,0) {};
\node[point,,label=above:{\scriptsize$K$}] (a6) at (7.5,0) {};
\draw[->,thick] (a1) -- (a2);
\draw[->,thick] (a2) -- (a3);
\draw[dashed,thick] (a3) -- (a4);
\draw[->,thick] (a4) -- (a5);
\draw[->,thick] (a5) -- (a6);
\node[point] (b) at (-3,-1) {};
\node[point] (b0) at (-1.5,-1) {};
\node[point,label=below:{\scriptsize$-K$}] (b1) at (0,-1) {};
\node[point] (b2) at (1.5,-1) {};
\node[point] (b3) at (3,-1) {};
\node[point] (b4) at (4.5,-1) {};
\node[point] (b5) at (6,-1) {};
\node[point,label=below:{\scriptsize$K$}] (b6) at (7.5,-1) {};
\node[point] (b7) at (9,-1) {};
\node[point] (b8) at (10.5,-1) {};
\draw[dashed] (b) -- ++(-1,0);
\draw[->] (b) -- (b0);
\draw[->] (b0) -- (b1);
\draw[->] (b1) -- (b2);
\draw[->] (b2) -- (b3);
\draw[dashed] (b3) -- (b4);
\draw[->] (b4) -- (b5);
\draw[->] (b5) -- (b6);
\draw[->] (b6) -- (b7);
\draw[->] (b7) -- (b8);
\draw[dashed] (b8) -- ++(1,0);
\draw[gray!50,line width=1mm,->] (a1) -- (b);
\draw[gray!50,line width=1mm,->] (a1) -- (b0);
\draw[gray!50,line width=1mm,->] (a1) -- (b1);
\draw[gray!50,line width=1mm,->] (a2) -- (b2);
\draw[gray!50,line width=1mm,->] (a3) -- (b3);
\draw[gray!50,line width=1mm,->] (a4) -- (b4);
\draw[gray!50,line width=1mm,->] (a5) -- (b5);
\draw[gray!50,line width=1mm,->] (a6) -- (b6);
\draw[gray!50,line width=1mm,->] (a6) -- (b7);
\draw[gray!50,line width=1mm,->] (a6) -- (b8);
\draw[rounded corners=3mm,<-,densely dotted,thick] (a1) -- ++(-1,0.5) -- ++(0,-1) -- (a1);
\draw[rounded corners=3mm,<-,densely dotted,thick] (a6) -- ++(1,0.5) -- ++(0,-1) -- (a6);
\node at (11,-0.6) {\large $\K$};
\node at (4,0.5) {minimal model of $\Sigma_\varphi$};
\end{tikzpicture}
\caption{The minimal model of $\Sigma_\varphi$ and $\K$.}\label{fig:horn}
\end{figure}
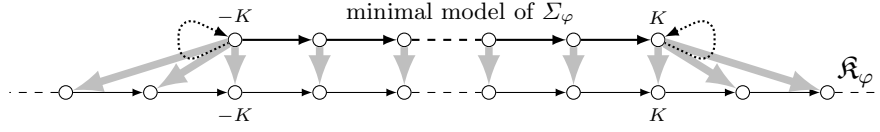

Obviously, $\varphi$ is satisfiable iff there is no negative clause $
\neg \lambda_1 \lor \dots \lor \neg \lambda_k$ in $\Phi^-$ such that all the $\lambda_i$
are true in $\K$ at some moment $n$ with $|n| \leq K$. This condition can be encoded by
means of propositional Horn clauses in the following way. 
For each variable $p$, we take $2K+1$ variables $p^n$,  $|n| \leq K$, and, for each $\Rbox p$ and $\Lbox p$, we take $2K+1$ variables, denoted $(\Rbox
p)^n$ and $(\Lbox p)^n$,  $|n| \leq K$,
respectively. Consider the following set $\Sigma_\varphi$ of
propositional Horn clauses, $|n| \leq K$:
\begin{align*}
\text{\iref{(H$_0$)}}\hspace*{1em}  & p^0, \quad \text{if } p\in\Psi,\\
\text{\iref{(H$_1$)}}\hspace*{1em}  & \lambda_1^n\land \dots\land \lambda_{k-1}^n \to \lambda_k^n, \quad\text{if } (\lambda_1\land \dots\land\lambda_{k-1} \to \lambda_k) \in \Phi^+,\hspace*{-12em} \\ 
%
%\text{\iref{(H$_1'$)}}\hspace*{1em}  & \neg p_1^n \lor \dots\lor \neg p_k^n, \quad\text{if } (\neg p_1 \lor \dots\lor \neg p_k) \in \Phi^-,\hspace*{-12em} \\ 
% 
\text{\iref{(H$_2$)}}\hspace*{1em}  & (\Rbox p)^n \to (\Rbox p)^{n+1}\ \ \  {\text{\footnotesize if $n < K$}}, && (\Lbox p)^n \to (\Lbox p)^{n-1} \ \ \  {\text{\footnotesize if $n > -K$}},\\ 
\text{\iref{(H$_3$)}}\hspace*{1em}  & (\Rbox p)^n \to p^{n+1}, && (\Lbox p)^n \to p^{n-1},\\ 
\text{\iref{(H$_4$)}}\hspace*{1em}  & (\Rbox p)^n \land p^n \to (\Rbox p)^{n-1}\  {\text{\footnotesize if $n > -K$}},%,\hspace*{-12em}\\
%& 
&&
(\Lbox p)^{n} \land p^{n} \to (\Lbox p)^{n+1} \ {\text{\footnotesize if $n < K$}},\\
\text{\iref{(H$_5$)}}\hspace*{1em}  & (\Rbox p)^{K} \leftrightarrow p^K, && (\Lbox p)^{-K} \leftrightarrow p^{-K},\\
\text{\iref{(H$_6$)}}\hspace*{1em} & (\Rbox p)^{-K} \leftrightarrow p^{-K}, && (\Lbox p)^{K} \leftrightarrow p^K.
\end{align*}
Clearly, $|\Sigma_\varphi| \le O(|\varphi|^2)$. It is readily seen that the minimal model of $\Sigma_\varphi$ corresponds to the canonical model $\K$ as shown in Fig.~\ref{fig:horn}. As propositional Horn satisfiability is \PTime-complete,  we obtain the following:
\begin{theorem}\label{thm:hornLTL}
The satisfiability problem for $\hornLTL\Xbox$-formulas is in \PTime.
\end{theorem}

\subsection{Temporal Derivations for $\coreLTL\Xbox$ in NLogSpace}\label{sec:nlogspace}

In $\coreLTL\Xbox$-formulas, all clauses are binary: $k = 2$ in~\eqref{eq:horn-clauses}. Satisfiability of propositional binary clauses is known to be \NLogSpace-complete. However, in the reduction $\varphi \mapsto \Sigma_\varphi$ above, the clauses \iref{(H$_4$)} are ternary. In this section we show how to modify the reduction to ensure membership in \NLogSpace{}.  More precisely,  we define two types of derivation from $\Psi\land\SVbox \Phi^+$: a 0-derivation of $(\lambda, n)$ will mean that $\K,n \models\lambda$, while a $\forall$-derivation of $\lambda$ from $\lambda'$ that $\K,0 \models \SVbox \lambda' \to \SVbox \lambda$. We then show that these derivations define $\K$ and that satisfiability of $\varphi$ can be checked  by a nondeterministic algorithm in  logarithmic space.

Denote by $\to^*$ the transitive and reflexive closure of the relation $\to$ over literals given by the clauses of $\Phi^+$. We require the following derivation rules over the pairs $(\lambda,n)$, where $\lambda$ is a positive temporal literal in $\varphi$ and $n\in\Z$:
\begin{align*}
\text{\iref{(R$_1$)}}\hspace*{1em} &  (\lambda_1, n) \Rightarrow (\lambda_2,n),\quad \text{if } \lambda_1 \to^* \lambda_2,\hspace*{-6em}\\ 
\text{\iref{(R$_2$)}}\hspace*{1em} & (\Rbox p, n) \Rightarrow (\Rbox p,n+1), && (\Lbox p, n) \Rightarrow (\Lbox p,n-1),\\
\text{\iref{(R$_3$)}}\hspace*{1em} & (\Rbox p, n) \Rightarrow (p,n+1), &&(\Lbox p, n) \Rightarrow (p,n-1),\\
\text{\iref{(R$_4$)}}\hspace*{1em} & (\Rbox p, 0) \Rightarrow (\Rbox p,-1), &&  (\Lbox p, 0) \Rightarrow (\Lbox p,1),\quad \text{if } p' \to^* p \text{ for } p'\in\Psi,\\ 
\text{\iref{(R$_5$)}}\hspace*{1em} & (p, n) \Rightarrow (\Rbox p,n-1), && (p, n) \Rightarrow (\Lbox p,n+1).
\end{align*}
The rules in \iref{(R$_1$)}--\iref{(R$_4$)} mimic~\iref{(H$_1$)}--\iref{(H$_4$)} above (\iref{(H$_4$)}  at moment 0 only) and reflect the semantics of \PTL{} in the sense that whenever $(\lambda,n) \Rightarrow (\lambda', n')$ and $\K,n \models \lambda$ then $\K,n' \models \lambda'$. 
For example, consider~\iref{(R$_4$)}. It only applies if $p$ follows (by $\to^*$) from the initial conditions in $\Psi$, in which case $\K, 0 \models p$, and so $\K, 0 \models \Rbox p$ implies $\K, -1 \models \Rbox p$. The rules in \iref{(R$_5$)} are different: for instance, we can only  apply $(p, n) \Rightarrow (\Rbox p,n-1)$ if we know that $p$ holds at all $m \ge n$. 

A sequence $\der\colon (\lambda_0,n_0) \Rightarrow \dots \Rightarrow  (\lambda_\ell,n_\ell)$, for $\ell \ge 0$, is called a \emph{$0$-derivation of $(\lambda_\ell,n_\ell)$}  if $\lambda_0 \in \Psi$, $n_0 =0$ and all applications of \iref{(R$_5$)} are \emph{safe} in the following sense: for any $(p, n_i) \Rightarrow_{\iref{(R$_5$)}} (\Rbox p,n_i-1)$, there is $\lambda_j = \Rbox q$, for some $q$ and $0 \le j < i$; similarly, for any $(p, n_i) \Rightarrow_{\iref{(R$_5$)}} (\Lbox p,n_i+1)$, there is $\lambda_j = \Lbox q$ with $0 \le j < i$. In this case we write $\Psi \Rightarrow^0 (\lambda_\ell,n_\ell)$.  For example, consider 
\begin{equation*}
\varphi \ \ = \ \ p \ \ \land \ \ \SVbox (p \to \Rbox q) \ \ \land \ \ \SVbox (q \to r) \ \ \land \ \ \SVbox (p \to r).
\end{equation*}
Evidently, $\K,-1 \models \Rbox r$. 
The following sequence is a 0-derivation of $(\Rbox r,-1)$  because the application of \iref{(R$_5$)} is safe due to $\Rbox q$:
\begin{equation*}
(p,0) \Rightarrow_{\iref{(R$_1$)}} (\Rbox q,0) \Rightarrow_{\iref{(R$_3$)}} (q,1) \Rightarrow_{\iref{(R$_1$)}} (r,1) \Rightarrow_{\iref{(R$_5$)}} (\Rbox r,0) \Rightarrow_{\iref{(R$_4$)}} (\Rbox r,-1).
\end{equation*}
Intuitively, if we can derive $(r,1)$ using $(\Rbox q,0)$, then we can also derive $(r,n)$ for any $n\geq 1$, and so we must also have $(\Rbox r, 0)$, which justifies the application of~\iref{(R$_5$)}. This argument is formalised in the following lemma:
\begin{lemma}[monotonicity]\label{lemma:monotone}
Let $\der$ be a 0-derivation of $(\lambda_\ell,n_\ell)$ with a suffix
\begin{equation}\label{eq:suff}
\mathfrak s \colon (\Rbox q, n_s) \Rightarrow (\lambda_{s+1},n_{s+1}) \Rightarrow \dots \Rightarrow (\lambda_\ell,n_\ell),
\end{equation}
where none of the $\lambda_i$ contains $\Rbox$. Then $\Psi\Rightarrow^0 (\lambda_\ell,m)$, for all $m \geq n_\ell$. Similarly, if there is a suffix beginning with some $\Lbox q$ then $\Psi\Rightarrow^0 (\lambda_\ell,m)$, for all $m \leq n_\ell$. Moreover, these 0-derivations only contain the  rules used in $\der$ and~\iref{(R$_2$)}.
\end{lemma}
\begin{proof}
We first remove all applications of~\iref{(R$_4$)} in~$\mathfrak{s}$. Let $(\lambda_i,n_i)\!\Rightarrow_{\iref{(R$_4$)}}\! (\lambda_{i+1},n_{i+1})$ be the first one.
By definition, $n_i = 0$ and, since $\Rbox q$ is the last $\Rbox$ in $\der$, we have $n_{i+1} = 1$ and  $\lambda_i = \lambda_{i+1} = \Lbox r$, for some $r$. So we can begin $\mathfrak s$ with $(\Rbox q, n_s) \!\Rightarrow_{\iref{(R$_2$)}}\! (\Rbox q, n_s\!+\!1)  
\Rightarrow (\lambda_{s+1},n_{s+1}\! +\! 1) \Rightarrow \!\cdots\! \Rightarrow (\lambda_i,n_i\!+\!1) \Rightarrow (\lambda_{i+2},n_{i+2})$; see Fig.~\ref{fig:monotonicity} on the left-hand side. We repeatedly apply this operation to obtain a suffix $\mathfrak{s}$ of the form~\eqref{eq:suff} that does not use~\iref{(R$_4$)}.
We then replace $\mathfrak{s}$ in $\der$ with $(\Rbox q, n_s)% \Rightarrow_{\iref{(R$_2$)}} (\Rbox q, m_0+1)  
\Rightarrow_{\iref{(R$_2$)}} \dots \Rightarrow_{\iref{(R$_2$)}} (\Rbox q, n_s+k)  
\Rightarrow (\lambda_{s+1},n_{s+1}+ k) \Rightarrow \dots \Rightarrow  (\lambda_\ell,n_\ell+k)$, where $k = m - n_\ell$; see Fig.~\ref{fig:monotonicity} on the right-hand side.
\qed
\end{proof}

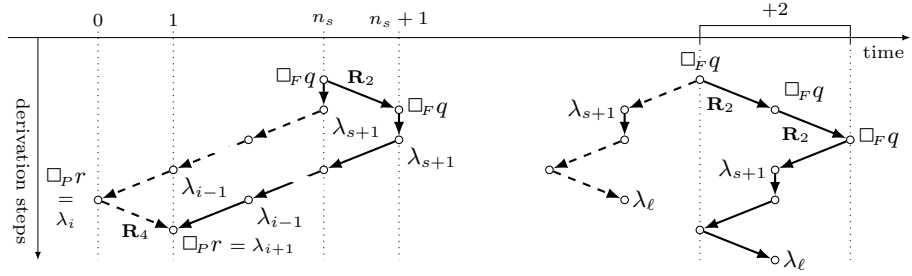
\begin{figure}[t]\centering
\begin{tikzpicture}[yscale=0.8,>=latex,point/.style={circle,draw=black,fill=white,minimum size=1mm,inner sep=0pt},label distance=-2pt]
\draw[ultra thin,->] (-3.8,0.7) to node [below,pos=0.6, sloped] {\scriptsize derivation steps} ++(0,-3.7);
\draw[ultra thin,->] (-4.2,0.7) to node [below,pos=0.97] {\scriptsize time} ++(12,0);
\draw[dotted] (-3,0.8) -- ++(0,-3.8);
\node at (-3,1) {\scriptsize 0};
\draw[dotted] (-2,0.8) -- ++(0,-3.8);
\node at (-2,1) {\scriptsize 1};
\draw[dotted] (0,0.8) -- ++(0,-3.8);
\node at (0,1) {\scriptsize $n_s$};
\draw[dotted] (1,0.8) -- ++(0,-3.8);
\node at (1,1) {\scriptsize $n_s+1$};
\node[point,label=left:{$\Rbox q$}] (q) at (0,0) {};
\node[point,label=below right:{\footnotesize$\lambda_{s+1}$}] (q1) at (0,-0.5) {};
\draw[thick,->,dashed] (q) to  (q1); % node[right] {\scriptsize\iref{R$_1$}}
\node[point] (q2) at (-1,-1) {};
\draw[thick,->,dashed] (q1) to (q2); %node[above] {\scriptsize\iref{R$_2$}/\iref{R$_3$}} 
\node[point,label=below right:{\footnotesize$\lambda_{i-1}$}] (q3) at (-2,-1.5) {};
\draw[thick,->,dashed] (q2) to  (q3); %node[above] {\scriptsize\iref{R$_2$}/\iref{R$_3$}}
\node[point,label=left:{\parbox{2em}{\centering$\Lbox r$\\\scriptsize $=$\\$\lambda_i$}}] (q4) at (-3,-2) {};
\draw[thick,->,dashed] (q3) to  (q4); %node[above] {\scriptsize\iref{R$_2$}/\iref{R$_3$}}
\node[point,label=below right:{$\Lbox r$\scriptsize\ $= \lambda_{i+1}$}] (q5) at (-2,-2.5) {};
\draw[thick,->,dashed] (q4) to node[below] {\scriptsize\iref{R$_4$}} (q5); %
\node[point,label=right:{$\Rbox q$}] (qp) at (1,-0.5) {};
\draw[thick,->] (q) to node[above] {\scriptsize\iref{R$_2$}} (qp); 
\node[point,label=below right:{\footnotesize$\lambda_{s+1}$}] (p1) at (1,-1) {};
\draw[thick,->] (qp) to  (p1); % node[right] {\scriptsize\iref{R$_1$}}
\node[point] (p2) at (0,-1.5) {};
\draw[thick,->] (p1) to (p2); %node[above] {\scriptsize\iref{R$_2$}/\iref{R$_3$}} 
\node[point,label=below right:{\footnotesize$\lambda_{i-1}$}] (p3) at (-1,-2) {};
\draw[thick,->] (p2) to  (p3); %node[above] {\scriptsize\iref{R$_2$}/\iref{R$_3$}}
\draw[thick,->] (p3) to  (q5); %node[above] {\scriptsize\iref{R$_2$}/\iref{R$_3$}}
\fill[white] (-1.5,-1.2) -- (-1.5,-1) -- (0,-1.7) -- (0,-1.9) -- cycle;
\draw[dotted] (5,0.6) -- ++(0,-3.6);
\draw[dotted] (7,0.6) -- ++(0,-3.6);
\draw[ultra thin] (5,0.6) to ++(0,0.3) to node[above,midway] {\scriptsize +2} ++(2,0) to ++(0,-0.3);
\node[point,label=above:{$\Rbox q$}] (s) at (5,0) {};
\node[point,label=left:{\footnotesize$\lambda_{s+1}$}] (s1) at (4,-0.5) {};
\draw[thick,->,dashed] (s) to  (s1); % node[right] {\scriptsize\iref{R$_1$}}
\node[point] (s2) at (4,-1) {};
\draw[thick,->,dashed] (s1) to (s2); %node[above] {\scriptsize\iref{R$_2$}/\iref{R$_3$}} 
\node[point] (s3) at (3,-1.5) {};
\draw[thick,->,dashed] (s2) to  (s3); %node[above] {\scriptsize\iref{R$_2$}/\iref{R$_3$}}
\node[point,label=right:{$\lambda_\ell$}] (s4) at (4,-2) {};
\draw[thick,->,dashed] (s3) to  (s4); %node[above] {\scriptsize\iref{R$_2$}/\iref{R$_3$}}
\node[point,label=above right:{$\Rbox q$}] (ts) at (6,-0.5) {};
\draw[thick,->] (s) to node[near start, below] {\scriptsize\iref{R$_2$}} (ts); 
\node[point,label=right:{$\Rbox q$}] (tss) at (7,-1) {};
\draw[thick,->] (ts) to node[near start, below] {\scriptsize\iref{R$_2$}} (tss); 
\node[point,label=left:{\footnotesize$\lambda_{s+1}$}] (t1) at (6,-1.5) {};
\draw[thick,->] (tss) to  (t1); % node[right] {\scriptsize\iref{R$_1$}}
\node[point] (t2) at (6,-2) {};
\draw[thick,->] (t1) to (t2); %node[above] {\scriptsize\iref{R$_2$}/\iref{R$_3$}} 
\node[point] (t3) at (5,-2.5) {};
\draw[thick,->] (t2) to  (t3); %node[above] {\scriptsize\iref{R$_2$}/\iref{R$_3$}}
\node[point,label=right:{$\lambda_\ell$}] (t4) at (6,-3) {};
\draw[thick,->] (t3) to  (t4); %node[above] {\scriptsize\iref{R$_2$}/\iref{R$_3$}}
\end{tikzpicture}
\caption{Removing applications of \iref{(R$_4$)} (left) and shifting a $0$-derivation by 2 (right): dashed arrows show the original derivation and solid ones   the resulting derivation.}\label{fig:monotonicity}
\end{figure}

However, 0-derivations are not enough to obtain all literals that are true in $\K$. Indeed, consider the formula
\begin{equation*}
\varphi \ \ = \ \ r \ \ \land \ \ \SVbox(r\to \Rbox q) \ \ \land \ \ \SVbox(\Rbox q \to q) \ \ \land \ \ \SVbox(\Lbox q \to p).
\end{equation*}
Clearly, $\K,n \models p$ for all $n \in \Z$, but neither $(p,n)$ nor $(\Lbox q,n)$ is 0-derivable. On the other hand, for each $n \in \Z$, there is a 0-derivation of $(q,n)$: for example,
\begin{equation*}
(r,0) \Rightarrow_{\iref{(R$_1$)}} (\Rbox q,0) \Rightarrow_{\iref{(R$_1$)}} (q,0) \Rightarrow_{\iref{(R$_5$)}} (\Rbox q, -1) \Rightarrow_{\iref{(R$_1$)}} (q, -1).
\end{equation*}
These 0-derivations correspond to $\K,0 \models \SVbox q$, from which we can derive $\SVbox p$ by means of the second type of derivations.  A sequence $\der\colon (\lambda_0,n_0) \Rightarrow \dots \Rightarrow  (\lambda_\ell,n_\ell)$ is called a \emph{$\forall$-derivation of $\lambda_\ell$ from $\lambda_0$} if it uses only~\iref{(R$_1$)}--\iref{(R$_3$)} and \iref{(R$_5$)}, whose applications are not necessarily safe.  
So we write $\Psi\Rightarrow^\forall  \lambda$ if  there is a $\forall$-derivation of $\lambda$ from some $q$ such that $\Psi\Rightarrow^0 (q,n)$, for all $n\in\Z$.
In the example above, $(q,0) \Rightarrow_{\iref{(R$_5$)}} (\Lbox q,1) \Rightarrow_{\iref{(R$_1$)}} (p,1)$ is a $\forall$-derivation of $p$ from $q$, whence $\Psi\Rightarrow^\forall p$.

\begin{lemma}[soundness]\label{sound} 
If $\Psi \Rightarrow^0 (\lambda,n)$  then $\K, n \models \lambda$. 
If $\Psi \Rightarrow^\forall \lambda$ then \mbox{$\K,0 \models \SVbox \lambda$}.
\end{lemma}
\begin{proof}
The proof for $\Rightarrow^0$ is by induction on the proof length. The basis of induction, $(\lambda,0)$ for $\lambda\in\Psi$, is by definition. Let $\der$ be a 0-derivation of $(\lambda,n)$. If the last rule application  is one of~\iref{(R$_1$)}--\iref{(R$_4$)} then $\K,n\models \lambda$ by the induction hypothesis. If $\der$ ends with $(p,n+1)\Rightarrow_{\iref{(R$_5$)}}(\Rbox p, n)$ then $\der$ without the last rule application contains $\Rbox$ and, by Lemma~\ref{lemma:monotone}, we obtain, for each $m \geq n+1$, a 0-derivation of $(p,m)$, whence $\K,n \models \Rbox p$.

The proof for $\Rightarrow^\forall$ is easy and left to the reader.\qed
\end{proof}

\begin{lemma}[completeness]\label{complete}
If $\K,n\models \lambda$ then either
$\Psi\Rightarrow^0 (\lambda,n)$
or $\Psi \Rightarrow^\forall \lambda$.
\end{lemma}
\begin{proof}
Let $\M$ be an interpretation such that, for all $p$ and $n\in\Z$, we have $\M,n\models p$ iff $\Psi\Rightarrow^0 (p,n)$ or $\Psi\Rightarrow^\forall p$. It suffices to show that $\M,0\models\Psi\land\SVbox \Phi^+$. Indeed, if we assume that there are $p'$ and $n'$ such that $\K,n'\models p'$ but neither $\Psi\Rightarrow^0(p',n')$ nor $\Psi\Rightarrow^\forall p'$, we will obtain $\M,n'\models\neg p'$ contrary to our assumption (other types of literals are considered analogously). 

Thus, we have to show that $\M$ is a model of $\Psi\land\SVbox \Phi^+$.
Suppose $p \in \Psi$. Then trivially $\Psi\Rightarrow^0 (p,0)$, and so $\M, 0 \models p$. Suppose $\lambda_1 \to \lambda_2 \in \Phi^+$ and $\M, n \models \lambda_1$. We  
consider three cases depending on the shape of $\lambda_1$ and show that $\M,n\models \lambda_2$.
\begin{description}
\item[\normalfont$\lambda_1 = p$.] If~$\Psi\Rightarrow^\forall p$ then, by~\iref{(R$_1$)}, $\Psi\Rightarrow^\forall\lambda_2$. Otherwise, there is a 0-derivation of $(p,n)$, and so $\Psi\Rightarrow^0 (\lambda_1,n) \Rightarrow_{\iref{(R$_1$)}}(\lambda_2,n)$. 

\item[\normalfont$\lambda_1= \Rbox p$.] Then $\M, m \models p$ for all $m > n$. Consider  $\M, n+1 \models p$. If $\Psi\Rightarrow^\forall p$ then, by~\iref{(R$_5$)},~\iref{(R$_1$)}, $\Psi\Rightarrow^\forall\lambda_2$.  Otherwise, there is a 0-derivation $\der$ of $(p,n+1)$.
\begin{description}
\item[\normalfont(F)] If $\Rbox$ occurs in $\der$ then $\Psi\Rightarrow^0 (p,n+1) \Rightarrow_{\iref{(R$_5$)}} (\Rbox p, n)\Rightarrow_{\iref{(R$_1$)}} (\lambda_2, n)$. 
\item[\normalfont(P)] If $\Lbox$ occurs in $\der$  then, by Lemma~\ref{lemma:monotone}, $\Psi \Rightarrow^0 (p,m)$ for each $m \leq n + 1$. Thus, $\Psi\Rightarrow^0 (p, m)$ for \emph{all} $m
  \in \Z$, and so, by~\iref{(R$_5$)} and~\iref{(R$_1$)}, $\Psi\Rightarrow^\forall \lambda_2$. 

\item[\normalfont(0)] If $\der$ contains neither $\Rbox$ nor $\Lbox$ then $n = -1$ and $\lambda \to^* p$, for some $\lambda\in\Psi$ (by~\iref{(R$_1$)}). %If $\M, 1 \models p$ is due to~\iref{(D$_2$)} then, by~\iref{(D$_2^\dagger$)}, $\M,n\models\lambda_2$. Otherwise, 
As $\M, 1 \models p$ and we assumed $\Psi\not\Rightarrow^\forall p$, there is a 0-derivation $\der'$ of $(p,1)$, which  must contain $\Rbox$ or $\Lbox$. If $\der'$ contains $\Rbox$ then $\Psi\Rightarrow^0 (p,1) \Rightarrow_{\iref{(R$_5$)}} (\Rbox p, 0)\Rightarrow_{\iref{(R$_4$)}} (\Rbox p, -1)\Rightarrow_{\iref{(R$_1$)}} (\lambda_2, n)$.  If $\Lbox$ occurs in $\der'$  then, by the argument in~(P), $\Psi\Rightarrow^\forall\lambda_2$.
\end{description}

\item[\normalfont$\lambda_1= \Lbox p$.] The proof is symmetric. 
\end{description}
In each of these cases, we have either $\Psi\Rightarrow^0 (\lambda_2,n)$ or $\Psi\Rightarrow^\forall \lambda_2$. Observe that $\Psi \Rightarrow^0 (\lambda_2,n)$ implies $\M,n\models\lambda_2$. Indeed, this clearly holds for $\lambda_2 = p$. If $\lambda_2 = \Rbox p$ then, by repetitive applications of~\iref{(R$_2$)} and an application of~\iref{(R$_3$)}, we obtain $\Psi\Rightarrow^0 (p,m)$, for all $m > n$, which means $\M,n\models\Rbox p$. The case $\lambda_2 = \Lbox p$ is symmetric. 
If $\Psi\Rightarrow^\forall \lambda_2$  then, independently of whether $\lambda_2$ is $p'$, $\Rbox p'$ or $\Lbox p'$, we have $\Psi\Rightarrow^\forall p'$, so $\M, m \models p'$ for all $m \in \Z$, whence, $\M,n \models \lambda_2$.\qed
\end{proof}

Next, in Lemmas~\ref{lem:short-chain} and~\ref{lem:univ-chain}, we provide efficient criteria for checking the conditions $\Psi\Rightarrow^0 (\lambda,n)$ and $\Psi \Rightarrow^\forall \lambda$ by restricting the range of numbers that can be used in 0-derivations (numbers in $\forall$-derivations can simply be ignored). Given a 0-derivation $\der\colon (\lambda_0,n_0) \Rightarrow \dots \Rightarrow  (\lambda_\ell,n_\ell)$, we define its \emph{reach} as
\begin{equation*}
r(\der)  = \max \{|n_i| \mid 0 \leq i \leq \ell\}.
\end{equation*}
We say that $\der$ \emph{right-stutters}, if there are $v<w$ such that  
$\lambda_v=\lambda_w$, $n_v < n_w$ and $n_i > 0$, for all $i$, $v \leq i
\leq w$ (in particular,~\iref{(R$_4$)} is not applied between $v$ and $w$). Symmetrically, $\der$ \emph{left-stutters} if there are $v < w$ such that $\lambda_v=\lambda_w$, $n_v > n_w$ and $n_i <0$, for all $i$, $v \leq i \leq w$. 

\begin{lemma}[checking~$\Rightarrow^0$]\label{lem:short-chain}
$\Psi\Rightarrow^0 (\lambda, n)$ iff there exists a $0$-derivation $\der$
  of $(\lambda, m)$ such that $r(\der) \le 2|\varphi|$ and one of the
  following conditions holds\textup{:}
\begin{description}
\item[\iref{(C$_1$)}] $m=n$\textup{;}
\item[\iref{(C$_2$)}] $\der$ contains $\Rbox$ and  either $m \leq n$ or $\der$ left-stutters\textup{;}
\item[\iref{(C$_3$)}] $\der$ contains $\Lbox$ and either $m \geq n$ or $\der$ right-stutters.
\end{description}
\end{lemma}
\begin{proof}
$(\Rightarrow)$ Let $\der\colon (\lambda_0,n_0) \Rightarrow \dots \Rightarrow (\lambda_\ell,n_\ell)$ be a 0-derivation of $(\lambda,n)$. If \mbox{$r(\der) \leq |\varphi|$} then $\der$ satisfies \iref{(C$_1$)}.
Otherwise, we take the first $\Box$-literal in $\der$, say $\lambda_t = \Rbox q$ (the case of $\Lbox q$ is symmetric). Clearly, $|n_t| \leq 1$. Let $u > t$ be the smallest index with  $|n_u| > |\varphi|$. Since adjacent $n_i$ and $n_{i+1}$ differ by at most 1, the segment between $(\lambda_t,n_t)$ and $(\lambda_u,n_u)$ contains a repeating literal: more precisely, there exist $v<w$ between $t$ and $u$ such that $\lambda_v=\lambda_w$ and 
\begin{itemize}
\item either $n_v > n_w \text{ and } n_i < 0 ,\text{ for } v \leq i \leq w$, 

\item or $n_v < n_w \text{ and } n_i > 0, \text{ for } v \leq i \leq
    w$.
\end{itemize}
In the former case $\der$ left-stutters, and we perform the following operations on the suffix $\mathfrak{s}\colon (\lambda_w, n_w) \Rightarrow \dots \Rightarrow (\lambda_\ell, n_\ell)$ of $\der$. First, we eliminate all applications of~\iref{(R$_4$)} in $\mathfrak{s}$: each suffix $(\Rbox q,0)\Rightarrow_{\iref{(R$_4$)}} (\Rbox q,-1) \Rightarrow (\lambda_s,n_s) \Rightarrow \dots \Rightarrow (\lambda_\ell,n_\ell)$ is replaced  by  $(\Rbox q,0) \Rightarrow (\lambda_s,n_s+1) \Rightarrow \dots \Rightarrow (\lambda_\ell,n_\ell+1)$; and 
similarly for $\Lbox$. If each time we eliminate the last application of~\iref{(R$_4$)} then the result is clearly a 0-derivation. Second, we remove all duplicating literals: each suffix $(\lambda_s,n_s) \Rightarrow \dots \Rightarrow (\lambda_{s'},n_{s'}) \Rightarrow (\lambda_{s'+1},n_{s'+1}) \Rightarrow  \dots \Rightarrow (\lambda_\ell,n_\ell)$ with $\lambda_s = \lambda_{s'}$ is replaced by $(\lambda_s,n_s) \Rightarrow  (\lambda_{s'+1},n_{s'+1} + k) \Rightarrow \dots \Rightarrow (\lambda_\ell,n_\ell + k)$, where $k = n_s - n_{s'}$.
%; that is, the segment from $s+1$ to $s'$ is removed and the tail is shifted by $k$.
%
%and we can construct a
%0-derivation $\der'$ by concatenating $(\lambda_0, n_0) \Rightarrow
%\dots \Rightarrow (\lambda_w, n_w)$ with the result of removing from
% all repetitions and adjusting the moments of time accordingly
%(for example, $(\lambda_{v+1}, n_{v+1}) \Rightarrow (q_1, 1)
%\Rightarrow (\Rbox q_2, 1) \Rightarrow (q_2, 2) \Rightarrow (\Rbox
%q_1, 2) \Rightarrow (q_1, 3) \Rightarrow (\Rbox p, 3) \Rightarrow
%(p,4)$ would result in $(\lambda_{v+1}, n_{v+1}) \Rightarrow (q_1, 1)
%\Rightarrow (\Rbox p, 1) \Rightarrow (p,2) $). 
This will give us a left-stuttering
0-derivation $\der'$ of $(\lambda, m)$, for some $m$. Since there are at most $|\varphi|$ distinct literals in $\mathfrak{s}$, we have $r(\der') \le
2|\varphi|$, thus satisfying the second option
of~\iref{(C$_2$)}; see Fig.~\ref{fig:stutter}.

\begin{figure}[t]\centering
\begin{tikzpicture}[yscale=0.75,>=latex,point/.style={circle,draw=black,fill=white,minimum size=1mm,inner sep=0pt},label distance=-2pt]
\draw[ultra thin,->] (-3.8,0.7) to node [below,pos=0.6, sloped] {\scriptsize derivation steps} ++(0,-4);
\draw[ultra thin,->] (-4.2,0.7) to node [below,pos=0.97] {\scriptsize time} ++(12,0);
\fill[gray!10] (0,0.7) rectangle (4,-3); 
\draw[dotted] (-3,0.8) -- ++(0,-3.8);
\node at (-3,1) {\scriptsize $m$};
\draw[dotted] (0,0.8) -- ++(0,-3.8);
\node at (0,1) {\scriptsize $-|\varphi|-1$};
\draw[dotted] (4,0.8) -- ++(0,-3.8);
\node at (4,1) {\scriptsize $-1$};
\draw[dotted] (5,0.8) -- ++(0,-3.8);
\node at (5,1) {\scriptsize $0$};
\draw[dotted] (6,0.8) -- ++(0,-4.3);
\node at (6,1) {\scriptsize $n_t$\tiny $({}=1)$};
\draw[dotted] (3,0.8) -- ++(0,-3.8);
\node at (3,1) {\footnotesize $n_v$};
\draw[dotted] (1,0.8) -- ++(0,-3.8);
\node at (1,1) {\footnotesize $n_w$};
\node[point,label=left:{$\lambda_0$}] (l0) at (5,0.2) {};
\node[point,label=right:{$\lambda_t$}] (lt) at (6,0) {};
\draw[thick,->,dashed] (l0) -- (lt);
\node[point,label=above:{\large $\lambda_v$}] (lv) at (3,-1) {};
\draw[thick,->,dashed] (lt) -- ++(-2,-0.4) -- ++(1,-0.2) -- (lv);
\node[point,label=below:{\large\hspace*{1em} $\lambda_w$\footnotesize ${}=\lambda_v$}] (lw) at (1,-2) {};
\draw[very thick,->,densely dotted] (lv) -- ++(-1,-0.2) -- ++(2,-0.4) -- (lw);
\node[point,label=left:{$\lambda$}] (l) at (-3,-3) {};
\draw[thick,->,dashed] (lw) -- ++(-2,-0.4) -- ++(1,-0.2) -- (l);
\end{tikzpicture}
\caption{Left-stuttering: $n_v$ and $n_w$ occur between $-1$ and $-|\varphi|-1$ (shaded) and the fragment of the derivation from $n_v$ to $n_w$ can be repeated any number of times (incl.\ $0$).}\label{fig:stutter}
\end{figure}
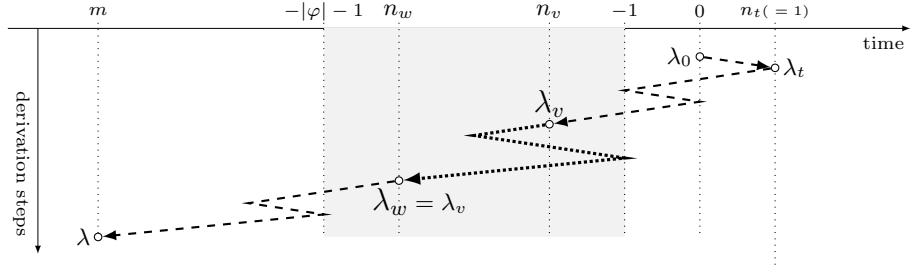

In the latter case $\der$ right-stutters, and we construct a
0-derivation $\der'$ of $(p, n')$ by cutting out the segment
$(\lambda_{v+1}, n_{v+1}) \Rightarrow \dots \Rightarrow (\lambda_w,
n_w)$ from $\der$ and `shifting' the tail using the construction above: eliminate applications of~\iref{(R$_4$)} and then decrease all numbers by $n_w - n_v > 0$. We then consider the obtained $\der'$ as the original
$\der$.
As the length of the derivations decreases and $n' \le n$, by applying this procedure 
sufficiently many times, we shall finally construct a 0-derivation of reach $\leq 2|\varphi|$ and satisfying either~\iref{(C$_1$)} or the first option of~\iref{(C$_2$)}.

\smallskip

$(\Leftarrow)$ is left to the reader. 
%We only show an example
%of how to derive $(p,n)$ for the case when the second option of
%\iref{(C$_2$)} holds. The derivation $(q_0, 0) \Rightarrow
%(\Rbox q_1, 0) \Rightarrow (q_2,0) \Rightarrow (\Rbox q_2, -1)
%\Rightarrow (q_3, -1) \Rightarrow (\Rbox q_3, -2) \Rightarrow (q_4,
%-2) \Rightarrow (\Rbox q_4, -3) \Rightarrow (q_3, -3) \Rightarrow (p,
%-3)$ left-stutters and satisfies \iref{(C$_2$)}. We can obtain, say, $(p, -4)$ using the derivation $(q_0, 0)
%\Rightarrow (\Rbox q_1, 0) \Rightarrow (\Rbox q_1, 1) \Rightarrow
%(q_2,1) \Rightarrow (\Rbox q_2, 0) \Rightarrow (q_3, 0) \Rightarrow
%(\Rbox q_3, -1) \Rightarrow (q_4, -1) \Rightarrow (\Rbox q_4, -2)
%\Rightarrow (q_3, -2) \Rightarrow (\Rbox q_3, -3) \Rightarrow (q_4,
%-3) \Rightarrow (\Rbox q_4, -4) \Rightarrow (q_3, -4) \Rightarrow (p,
%-4)$. This construction easily generalises to $(p, n)$, for any
%$n \in \Z$. 
\qed
\end{proof}

In a similar way we can show how to efficiently check the condition $\Psi \Rightarrow^\forall p$:
\begin{lemma}[checking~$\Rightarrow^\forall$]\label{lem:univ-chain}
 $\Psi\Rightarrow^0 (\lambda, n)$ holds for all $n \in \Z$ iff there are
  $0$-derivations $\der$ of $(\lambda,m)$ and $\der'$ of $(\lambda, m')$ of reach
  at most $2|\varphi|$ such that one of the following conditions
  holds\textup{:}
\begin{description}
  \item[\iref{(C$_1'$)}] $\der$ contains $\Rbox$, $\der'$ contains $\Lbox$ and $m \leq m'+1$\textup{;}
  \item[\iref{(C$_2'$)}] $\der$ contains $\Rbox$ and left-stutters\textup{;}
  \item[\iref{(C$_3'$)}] $\der$ contains $\Lbox$ and right-stutters.
  \end{description}
\end{lemma}
\begin{proof}
$(\Rightarrow)$ Take a 0-derivation of $(q, 2|\varphi|+1)$. By
Lemma \ref{lem:short-chain}, there is a derivation $\der_0$ of
$(q,n_0)$ with $r(\der_0)\le 2|\varphi|$ satisfying either
\iref{(C$_2$)} or \iref{(C$_3$)}. If $\der_0$ left- or right-stutters then we have \iref{(C$_2'$)} or  \iref{(C$_3'$)}, respectively. Otherwise, $\der_0$ contains $\Rbox$ and we can construct a finite sequence of 0-derivations $\der_0, \der_1, \der_2,\dots,\der_k$ of reach at most $2|\varphi|$, where each $\der_i$ is a 0-derivation of $(q,n_i)$ containing $\Rbox$,  and such that $n_0 > n_1 > n_2 > \dots > n_k$.

Suppose we have already constructed $\der_i$.  Since $\Psi\Rightarrow^0 (q, n)$, for all $n$, we have $\Psi\Rightarrow^0 (q, n_i-1)$. By Lemma \ref{lem:short-chain},  there is a 0-derivation $\der$ of $(q, n_{i+1})$, for some $n_{i+1}$, with one of \iref{(C$_1$)}--\iref{(C$_3$)}.   If \iref{(C$_2$)} and $\der$ left-stutters or \iref{(C$_3$)} and $\der$ right-stutters then we obtain \iref{(C$_2'$)} or \iref{(C$_3'$)}, respectively.  If~\iref{(C$_2$)} and $\der$ contains $\Rbox$ with $n_{i+1} \leq n_i - 1$ then $\der$ becomes the next member $\der_{i+1}$ in the sequence.  If \iref{(C$_3$)} and $\der$ contains $\Lbox$ with $n_{i+1} \geq n_i - 1$ then $\der_i$ and $\der$ satisfy \iref{(C$_1'$)}. Otherwise, we have~\iref{(C$_1$)} with $n_{i+1} = n_i - 1$ (recall that $n_i >  -2|\varphi|$). Consider three cases. If $\der$ contains $\Rbox$ then $\der$ becomes the next member $\der_{i+1}$ in the sequence. If $\der$ contains $\Lbox$ then $\der_i$ and $\der$ satisfy \iref{(C$_1'$)}.   Otherwise, that is, if $\der$ contains neither $\Lbox$ nor $\Rbox$, we  must  have $n_{i+1} = 0$ and $p \to^* q$, for some $p\in\Psi$. Then we have $n_i = 1$ and, as $\der_i$ contains $\Rbox$, we can append $(q,1) \Rightarrow_{\iref{(R$_5$)}} (\Rbox q, 0)\Rightarrow_{\iref{(R$_4$)}}(\Rbox q, -1)\Rightarrow_{\iref{(R$_3$)}} (q, 0)$ to $\der$ to obtain the next member $\der_{i+1}$ in the sequence.

\smallskip
 
$(\Leftarrow)$ is left to the reader.    
\qed
\end{proof}

We are now in a position to prove the main result of this section.

\begin{theorem}\label{newstuff}
The satisfiability problem for $\coreLTL\Xbox$-formulas is in \NLogSpace.
\end{theorem}
\begin{proof}
An $\coreLTL\Xbox$-formula $\varphi = \Psi\land \SVbox \Phi^+ \land \SVbox \Phi^-$  
%of the form $\Psi\land \SVbox \Phi^+ \land \SVbox \Phi^-$. 
is unsatisfiable 
iff $\Phi^-$ contains a clause
$\neg \lambda_1 \lor \neg \lambda_2$ such that $\K, n \models \lambda_1\land \lambda_2$, for
some $n$ with $|n| \le K$. 
For each  $\neg \lambda_1 \lor \neg \lambda_2$ in $\Phi^-$, our algorithm guesses such an $n$ (in binary) and, for both $\lambda_1$ and $\lambda_2$, checks whether $\Psi\Rightarrow^0 (\lambda_i,n)$ or 
$\Psi \Rightarrow^\forall \lambda_i$, which, by Lemmas~\ref{lem:short-chain} and~\ref{lem:univ-chain},  requires only logarithmic space. 
\qed
\end{proof}

The initial clauses of $\coreLTL\Xbox$-formulas $\varphi$ are propositional variables.
If we slightly extend the language to allow for initial core-clauses (without $\SVbox$), 
then the satisfiability problem becomes \PTime-hard. This can be shown by reduction of satisfiability of
propositional Horn formulas with clauses of the form $p$, $\neg p$ and
$p\land q\to r$, which is known to be \PTime-complete. Indeed, suppose
$f=\bigwedge_{i=1}^n C_i$ is such a formula. We define a temporal
formula $\varphi_f$ to be the conjunction of all unary clauses of $f$
with the following formulas, for  each ternary clause $C_i$ of the form
$p \land q \to r$ in $f$:
\begin{equation*}
c_i \ \ \land \ \ \SVbox (p \to \Rbox c_i) \ \ \land \ \ \SVbox (q \to \Lbox c_i)  \ \ \land \ \  (\SVbox c_i \to r),
\end{equation*}
where $c_i$ is a fresh variable.
One can show that $f$ is
satisfiable iff $\varphi_f$ is satisfiable.

%***********************

We finish this section by an observation that if the language allows for non-Horn clauses (e.g., $p \lor q$) then the satisfiability problem  becomes \NP-hard: 
 \begin{theorem}\label{krom-low-NP}
The satisfiability problem for $\kromLTL\Xbox$-formulas is \NP-hard.
\end{theorem}
\begin{proof}
By reduction of graph
3-colourability. Given a graph $G =(V,E)$,
consider the following $\kromLTL\Xbox$-formula $\varphi_G$ with
variables $p_0,\dots, p_4$ and $\overline{v}_i$, for $v_i\in V$:
\begin{multline*}
  p_0 \ \  \land \ \  \bigwedge\nolimits_{0 \leq i \leq 3} \SVbox (p_i \to \Rbox
  p_{i+1}) \  \ \land \ \ 
  \bigwedge\nolimits_{v_i\in V} \SVbox(p_0 \to \neg \Rbox
   \overline{v}_i) \ \ \land {} \\ \bigwedge\nolimits_{v_i\in V} \SVbox (p_4 \to  \overline{v}_i) \ \ \land \ \ 
    \bigwedge\nolimits_{(v_i,v_j)\in E} \SVbox(\overline{v}_i \lor
   \overline{v}_j).
 \end{multline*}%
Intuitively, the first four conjuncts of this formula  choose, for each vertex $v_i$ of the graph, a moment of time $1 \le n_i \leq 3$; the last conjunct makes sure that $n_i \ne n_j$ in case $v_i$ and $v_j$ are connected by an edge in $G$. We claim that $\varphi_G$ is satisfiable iff $G$ is 3-colourable. If $c\colon V \to \{1,2,3\}$ is a colouring of $G$ then set $\M,n \models  \overline{v}_i$ iff $c(v_i) \neq n$, for $v_i \in V$, and $\M,n \models p_i$ just in case $n\geq i$, for each $p_i$.
%$0 \leq i \leq 3$; and  $\M,n \models p_4$ iff $n> 3$.
Clearly, $\M,0 \models \varphi_G$. Conversely, if $\M,0\models\varphi_G$ then, for each $v_i \in V$, there is $n_i \in \{1,2,3\}$ with $\M,n_i \models \neg \overline{v}_i$ and $\M,n_i
\models \overline{v}_j$ whenever $(v_i,v_j) \in E$. Thus, $c \colon
v_i \mapsto n_i$ is a colouring of $G$.
\qed
\end{proof}

%**********************

\section{Conclusion}

We have investigated the computational complexity of the satisfiability problem for the fragments of \PTL{} over $(\Z,<)$ given by the form of the clauses---\textit{bool}, \textit{horn}, \textit{krom} and \textit{core}---in the clausal normal form and the temporal operators available for constructing temporal literals. Apart from $\boolLTL\Xallop$, whose formulas are equisatisfiable to formulas in the full \PTL, only $\hornLTL\Xallop$ has \PSpace-complete satisfiability. For all other fragments, the complexity varies from \NLogSpace{} to  \PTime{} and \NP.

The idea to consider sub-Boolean fragments of \PTL{} comes from description logic, where the {\sl DL-Lite}  family~\cite{CDLLR07,ACKZ:jair09} of logics has been designed and investigated with the aim of finding formalisms suitable for ontology-based data access (OBDA). It transpired that, despite their low complexity, {\sl DL-Lite} logics were capable of representing basic conceptual data modelling constructs~\cite{BeCD05-AIJ-2005,ACKRZ:er07}, and gave rise to the W3C standard ontology language {\sl OWL~2~QL} for OBDA. One possible application  of the results obtained in this paper lies in temporal conceptual modelling and temporal OBDA~\cite{AKRZ:ER10}. Temporal description logics (and  other many-dimensional logics) are notorious for their bad computational properties~\mbox{\cite{GKWZ03,LuWoZa-TIME-08}.} We believe, however, that efficient practical reasoning can be achieved by considering sub-Boolean temporal extensions of {\sl DL-Lite} logics; see~\cite{AKRZ:aaai10} for first promising results.

%%
%% Bibliography
%%

%\bibliography{biblio}

\end{document}